\documentclass[a4paper,11pt,oneside]{report}

\usepackage[
backend=biber,
style=alphabetic,
sorting=nty,
maxbibnames=999
]{biblatex}
\usepackage[utf8]{inputenc}
\usepackage{tikz}
\usetikzlibrary{arrows}
\usepackage{float}
\usepackage[utf8]{inputenc}
\usepackage[english]{babel}
\usepackage[euler]{textgreek}
\usepackage{titlesec}
\usepackage{amsthm}
\usepackage{subcaption}
\usepackage{mwe}
\usepackage{graphicx}
\usepackage{hyperref}
\usepackage{siunitx}
\usepackage[export]{adjustbox}
\usepackage{bm}
\usepackage{chngcntr}
\counterwithout{figure}{chapter}
\usepackage{cleveref}
\usepackage{lipsum}
\usepackage{xcolor}
\usepackage{float}
\usepackage{moresize}
\usepackage{datetime}
\usepackage{xurl}
\newdateformat{monthyeardate}{%
\monthname[\THEMONTH], \THEYEAR}

\usepackage{outline} \usepackage{pmgraph} \usepackage[normalem]{ulem} \usepackage{verbatim}

\def \Dj{\mbox{\raise0.3ex\hbox{-}\kern-0.4em D}} 

\begin{titlepage}
\title{
\includegraphics[width=2.25in]{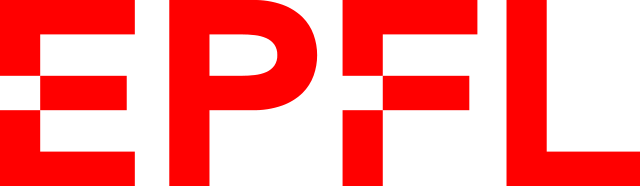} \\
\vspace*{0.7in}
\huge Master Thesis \\
\vspace*{0.3in}
\textbf{\HUGE Double Choco is NP-complete}}
\author{\huge Dragoljub \Dj urić \\
        \vspace*{0.5in} \\
        Thesis Advisor: Prof. Dr. Mika Göös\\ 
        \vspace*{2in} \\
		School of Computer and Communication Sciences\\
		Theory of Computation Laboratory 5\\
        \textbf{École polytechnique fédérale de Lausanne}\\
        Lausanne, Switzerland} 
        \date{February, 2022}
        
\end{titlepage}

\newtheorem{theorem}{Theorem}
\theoremstyle{definition}
\newtheorem{definition}{Definition}
\newtheorem{claim}{Claim}
\newtheorem{lemma}{Lemma}
\crefname{lemma}{Lemma}{Lemmas}

\newcommand{\squareop}[1]{%
    \setlength{\fboxsep}{0pt}%
    \setlength{\unitlength}{.7em}%
    \mathrel{%
        \raisebox{-1pt}{\framebox(1,1){\(\scriptstyle #1\)}}%
    }%
}

\begin{document}
    \maketitle
    
    \addtocounter{page}{1}

    \tableofcontents
     \addtocounter{page}{1}
    \newpage
    \hfill \textcolor{gray}{\textit{'You can't connect the dots looking forward; you can only connect them looking backward.'}}
    
    \hfill \textcolor{gray}{S. J.}
    
    \newpage
    \chapter{Introduction} \label{chap:1}
    
    Nikoli is a leading publisher of pencil-and-paper logic puzzles, having presented many different puzzles over the past few decades in their \textit{Puzzle Communication Nikoli} magazine and website \cite{nikoli}. Perhaps the most well known of their brain teasers is \textit{Sudoku} \cite{sudoku}, which was proven to be \textit{NP-complete} in the year 2003 \cite{yato2003complexity}.
    
    As Nikoli expanded into the field of puzzles, their puzzles and the study of their computational properties have gained an increasing amount of popularity. Consequently, the concept of the \textit{Nikoli gap} \cite{nicoligap} was established to define the number of years between the year a puzzle is introduced and the year a paper about the puzzle`s computational complexity is published. 
    
    \begin{figure}[htb]
    \begin{tikzpicture}
      \node  {\centerline{\includegraphics[scale=0.37]{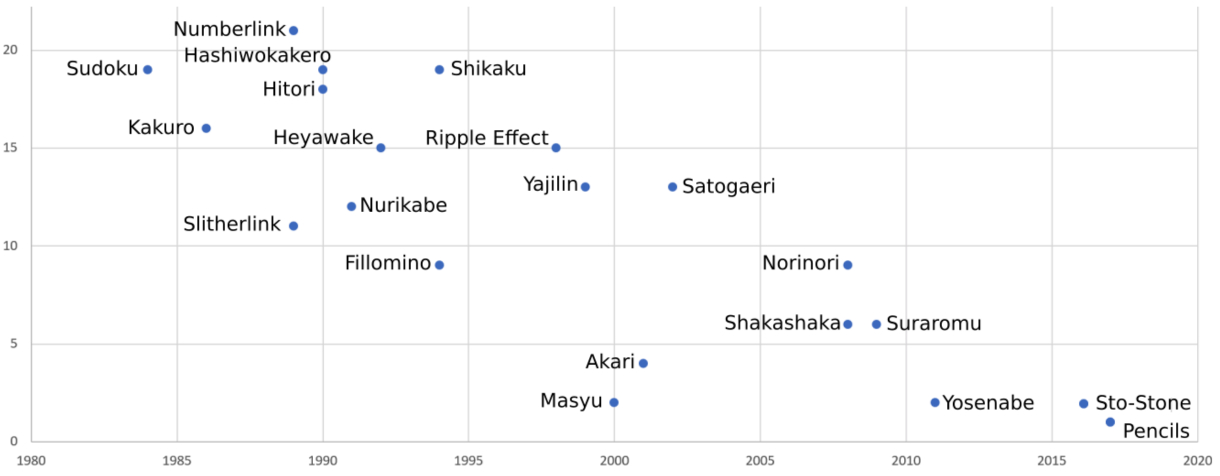}}};
      \node[node distance=0cm, yshift=-3.5cm,font=\color{gray}] {Year the puzzle is introduced};
      \node[node distance=0cm, rotate=90, anchor=center,yshift=8.2cm,font=\color{gray}] {Nikoli gap};
     
    \end{tikzpicture}
    \caption{Nikoli gap. Source: \cite{nicoligap}}
    \label{GAP}
    \end{figure}
    
    Papers on the computational complexity of the puzzles appearing in \autoref{GAP} can be found in: Kakuro (Cross Sum) \cite{yato2003complexity}, Slitherlink \cite{takayuki2000np}, Numberlink \cite{kotsuma2010np}, Hashiwokakero \cite{andersson2009hashiwokakero}, Hitori \cite{hearn2009games}, Nurikabe \cite{holzer2004np}, Heyawake \cite{holzer2007troubles}, Fillomino \cite{fillomino}, Shikaku and Ripple Effect \cite{takenaga2013shikaku}, Yajilin \cite{ishibashi2012np}, Masyu (Pearl) \cite{friedman2002pearl}, Akari (Light Up) \cite{mcphail2005light}, Satogaeri \cite{kanehiro2015satogaeri}, Shakashaka \cite{adler2015computational, demaine2014computational}, Norinori \cite{biro2017computational}, Suraromu \cite{kanehiro2015satogaeri}, Yosenabe \cite{iwamoto2014yosenabe}, Sto-Stone \cite{allen2018sto, nicoligap}, and Pencils \cite{nicoligap}.
    
    We can observe that the \textit{Nikoli gap} has a decreasing trend that can be justified by the rise in popularity of \textit{Nikoli} puzzles with customers as well as an increase in focus within the theoretical computer science community. This has been supported by continued acceptance of papers regarding puzzles' computational properties on conferences such as \textit{FUN} \cite{fun}. 

    In this thesis, we are studying the computational properties of the \textbf{Double Choco} (Double Chocolate) puzzle \cite{dchoco}, which was first published on January 15, 2020 in \cite{dcpub}, gaining huge popularity in a relatively short period of time. \textit{Double Choco} is a Nikoli  pencil-and-paper game consisting of an $m\times n$ grid board of unit squares, of white or gray color, separated by dotted lines. Some squares may contain an integer between $1$ and $\frac{m n}{2}$. The solution to this puzzle is a set of blocks such that:
    \begin{itemize}
        \item The blocks are disjointed.
        \item The blocks, together, cover all cells of the puzzle.
        \item Each block has the same number of white and gray cells.
        \item The white and gray areas of any block are connected (both separately and with each other) and they have the same size and shape, with respect to rotation or mirroring.
        \item If any cell inside a block contains a number, then the white and gray area of that block must have that number of cells.
        \item A block can contain any number of cells with numbers.
    \end{itemize}

    \begin{figure}[H]
        \centering
        \begin{subfigure}[b]{0.22\textwidth}
            \centering
            \includegraphics[width=\textwidth]{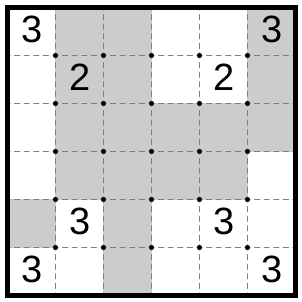}
            \caption[Network2]%
            {{\small }}    
            \label{fig:mean and std of net14}
        \end{subfigure}
        \hfill
        \begin{subfigure}[b]{0.22\textwidth}
            \centering
            \includegraphics[width=\textwidth]{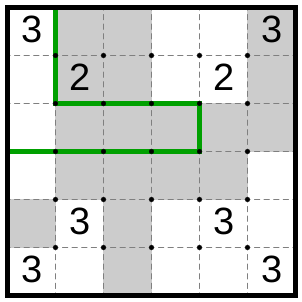}
            \caption[Network2]%
            {{\small }}    
            \label{fig:mean and std of net14}
        \end{subfigure}
        \hfill
        \begin{subfigure}[b]{0.22\textwidth}
            \centering
            \includegraphics[width=\textwidth]{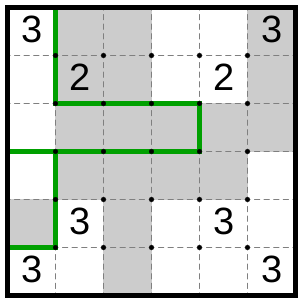}
            \caption[Network2]%
            {{\small }}    
            \label{fig:mean and std of net14}
        \end{subfigure}
        \hfill
        \begin{subfigure}[b]{0.22\textwidth}  
            \centering 
            \includegraphics[width=\textwidth]{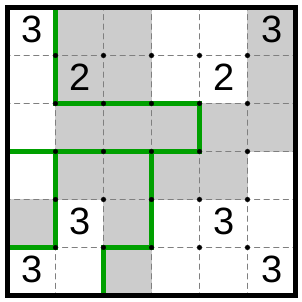}
            \caption[]%
            {{\small }}    
            \label{fig:mean and std of net24}
        \end{subfigure}
        \vskip\baselineskip
        \begin{subfigure}[b]{0.22\textwidth}
            \centering
            \includegraphics[width=\textwidth]{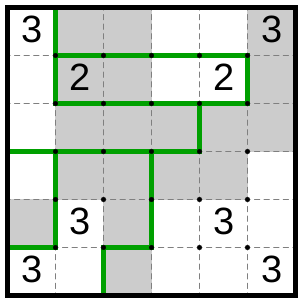}
            \caption[Network2]%
            {{\small }}    
            \label{fig:mean and std of net14}
        \end{subfigure}
        \hfill
        \begin{subfigure}[b]{0.22\textwidth}
            \centering
            \includegraphics[width=\textwidth]{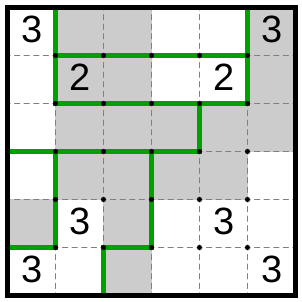}
            \caption[Network2]%
            {{\small }}    
            \label{fig:mean and std of net14}
        \end{subfigure}
        \hfill
        \begin{subfigure}[b]{0.22\textwidth}   
            \centering 
            \includegraphics[width=\textwidth]{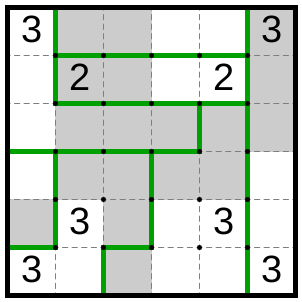}
            \caption[]%
            {{\small }}    
            \label{fig:mean and std of net34}
        \end{subfigure}
        \hfill
        \begin{subfigure}[b]{0.22\textwidth}   
            \centering 
            \includegraphics[width=\textwidth]{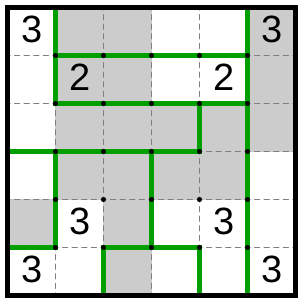}
            \caption[]%
            {{\small }}    
            \label{fig:mean and std of net44}
        \end{subfigure}
        \caption[ The average and standard deviation of critical parameters ]
        {\small A step by step solution of the Double Choco puzzle from \textit{The Guardian} article \cite{guardian}.}
        \label{fig:puzzle-inst}
    \end{figure}
    
    We create the puzzles using Double Choco editor \cite{dceditor}. The editor has shown to be a very useful tool for constructing and solving the puzzles as it is capable of automatically checking the solution.
    
    In \autoref{fig:puzzle-inst}, we present the puzzle instance with the step-by-step solution. In each step, one block is chosen, so the number of steps is bounded by the number of cells (more precisely by its half), so the problem is in \textit{NP} (see \autoref{sect:NP}). In this thesis, we prove that it is in fact \textit{NP-complete}.
    
    We start by introducing the necessary tools used in our study in \autoref{chap:2}. We first explain what  \textit{NP-completeness} is, followed by the introduction of the \textit{Planar Rectilinear Monotone 3SAT} problem that we use to prove the hardness of our result. Finally, we conclude the chapter with the properties our reductions should satisfy.
    
    The first important result of this thesis is presented in \autoref{chap:3} where we prove that the arbitrarily shaped \textit{Double Choco} is \textit{NP-complete}.
    
    Finally, in \autoref{chap:4}, by using the result from the previous chapter, we prove that \textit{Double Choco} is \textit{NP-complete}, hence establishing the 
    \textit{Nikoli gap} of more than 2 years.

    \chapter{Preliminaries} \label{chap:2}
    
    \section{NP-completeness} \label{sect:NP}
    This section is mainly taken from \cite{fillament}.
    The complexity class \textit{NP} (Nondeterministic Polynomial time) represents a set of decision
    problems, i.e., problems with yes-or-no answers, such that each solution has a polynomial length with
    respect to input size. Moreover, there exists a program that can check the validity of such a
    solution (witness) in polynomial time.

    More formally:
    \theoremstyle{definition}
    \begin{definition}[\textit{NP}]
    NP is the class of problems $A$ of the form:\newline
    $x$ is a yes-instance of $A$ if, and only if, there exists a $w$ such that $(x, w)$ is a
    yes-instance of B, where B is a decision problem in \textit{P} regarding pairs $(x, w)$, and $|w| =
    poly(|x|)$ \cite{moore2011nature}.
    \end{definition}
    
    Reduction is a procedure for transforming an instance of one problem into an instance of
    the another problem. If a problem $A$ is reducible to problem $B$, and if there exists an efficient algorithm for solving problem $B$, then we can use it to solve $A$ efficiently as well, i.e., the problem $B$ is at least as hard as problem $A$. It this thesis we will only use \textit{mapping reduction}:
    
    \theoremstyle{definition}
    \begin{definition}[\textit{Mapping reduction}]
    A function f: $\{0, 1\}^{*}\ \rightarrow \{0, 1\}^{*}$ is called mapping reduction from $A$ to $B$ if, and only if,:
    \begin{enumerate}
         \item For any $w\ \in \textSigma_{1}*$, $w \in A$ if, and only if, $f(w) \in B$.
         \item f is a computable function, and can be performed in polynomial time.
    \end{enumerate}

   Or less formally, mapping reduction from $A$ to $B$ says that we can transform any instance of $A$ into an instance of $B$, such that the answer to $B$ is the answer to $A$.
    
    \end{definition}
    
    \theoremstyle{definition}
    \begin{definition}[NP-hard]
    A problem is \textit{NP-hard} if every problem in NP can be reduced to it using at most polynomial time
    reduction.
    \end{definition}
    
    A problem that is both \textit{NP} and \textit{NP-hard} is \textit{NP-complete}, i.e., if there exists an algorithm that
    solves any of \textit{NP-complete} problems in polynomial time, we can easily solve all of them in
    polynomial time. In that case, it will hold that $P = NP$, but it is highly unlikely, so we believe
    that $P \neq NP$, i.e., we believe that \textit{P} and \textit{NP} are two essentially different complexity
    classes.

    \section{Rectilinear Planar Monotone 3SAT} \label{sec:rect}
    
    A Boolean variable $x_i$ can have a value of \textit{true} or \textit{false}. The variable or its negation are collectively called literal. A Boolean formula $\phi$ consists of the Boolean variables $x_1$, $x_2$, ..., $x_n$ and the logical operators AND ($\wedge$), OR ($\vee$), and NOT ($\neg$), e.g. $\phi = (x_1 \wedge x_2) \vee (\neg x_2 \vee x_3)$. A formula $\phi$ is satisfiable if there exists assignment of the variables $x_1$, $x_2$, ..., $x_n$ such that $\phi(x_1, ..., x_n) = true$, otherwise, it is unsatisfiable \cite{arora2009computational}.
    
    A Boolean formula is in the \textit{Conjunctive Normal Form} (\textit{CNF}) if it consists of groups of OR's of literals (clauses), with AND's between clauses. A \textit{CNF} formula having at most $3$ literals in each clause is known as an \textit{3CNF} formula.
    
    The \textit{Boolean satisfiability problem} (\textit{SAT}) is the problem of deciding whether a given \textit{CNF} formula has a satisfying assignment. Similarly, for the \textit{3CNF} formula is the \textit{3SAT}. The fundamental result of the complexity theory known as the \textit{Cook-Levin Theorem} says that the \textit{SAT} (\textit{3SAT}) is \textit{NP-complete} \cite{cook1971complexity, levin1973universal}.
    
    Here, we will focus on the restricted version of \textit{3SAT} known as the \textit{Rectilinear Planar Monotone 3SAT} (\textit{RPM-3SAT}) (\autoref{fig:RPM-planar}). A formula is planar if the incidence graph of variables and clauses, where a variable is connected to a clause containing it, can be drawn on the plane in a such a way that none of its edges intersect (i.e. they can be embedded in the plane). Rectilinear property ensures that planar embedding satisfies:
    
    \begin{itemize}
        \item Variables are horizontal rectangles of the same height.
        \item Clauses are horizontal rectangles of the same height.
        \item Variables are connected to clauses with the vertical lines (wires).
        \item There exists a \textit{variable line} - on which all variables lie.
    \end{itemize}
    
    Monotonicity implies that each clause is either positive or negative, meaning they contain only the variables without negations or with negations respectively. Hence, adding two properties:
    
    \begin{itemize}
        \item Positive clauses are above the \textit{variable line}.
        \item Negative clauses are below the \textit{variable line}.
    \end{itemize}
    
        \begin{figure*}[t]
                \centering
                \begin{subfigure}[b]{0.40\textwidth}
                    \centering
                    \includegraphics[width=\textwidth]{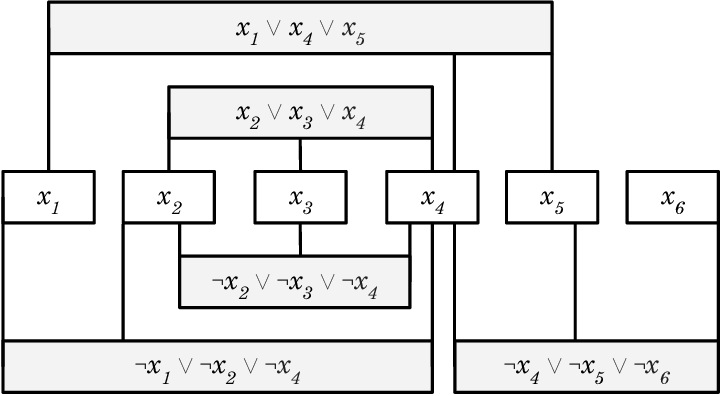}
                    \caption[Network2]%
                    {{\small Rectilinear embedding.}}    
                    \label{fig:RPM-planar}
                \end{subfigure}
                \hfill
                \begin{subfigure}[b]{0.40\textwidth}  
                    \centering 
                    \includegraphics[width=\textwidth]{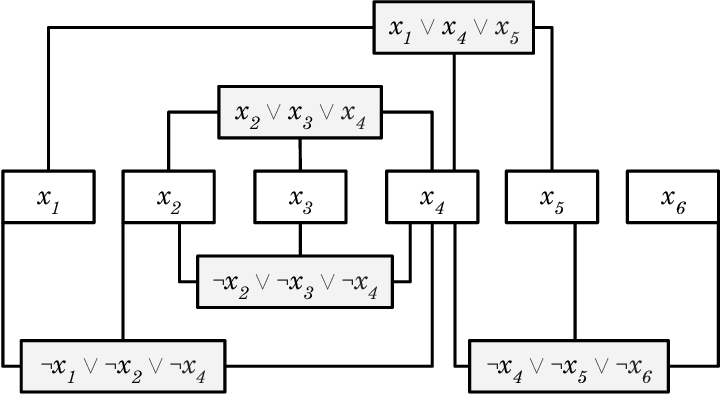}
                    \caption[]%
                    {{\small Bent rectilinear embedding.}}    
                    \label{fig:bent}
                \end{subfigure}
            \caption[ The average and standard deviation of critical parameters ]
            {\small Embeddings of the instance of RPM-3SAT. They represent the CNF formula: $(x_1 \vee x_4 \vee x_5) \wedge (x_2 \vee x_3 \vee x_4) \wedge (\neg x_1 \vee \neg x_2 \vee \neg x_4) \wedge (\neg x_2 \vee \neg x_3 \vee \neg x_4) \wedge (\neg x_4 \vee \neg x_5 \vee \neg x_6)$. Source: \cite{allen2018sto}} 
            \label{fig:RPM}
        \end{figure*}
        
    \begin{theorem}
        Rectilinear Planar Monotone 3SAT is \textit{NP-complete}.
    \end{theorem}
    \begin{proof}
        Look at \cite{de2012optimal}.
    \end{proof}
    
    In order to simplify our reduction, we introduce \textit{Bent RPM-3SAT} by making a small adjustment to planar embedding as shown in \autoref{fig:bent} (inspired by \cite{allen2018sto}):
    \begin{itemize}
        \item Each clause is shrunk by moving the left and right boundary toward the clause center, in such a way that all clauses have the same width and that the middle wire is equally distanced between the left and the right end of clause.
        \item Each left and right wire of every clause is extended vertically from the bottom to middle of the clause, redirected by $90\si{\degree}$ in clause direction, and extended horizontally until it touches the clause border.
    \end{itemize}
    
    Since the \textit{RPM-3SAT} is drawn on a plane (continuous domain), each of its clauses can be shrunk and wires can be modified, as explained, without changing any other part of the embedding. Hence, the \textit{Bent RPM-3SAT} will preserve the computational properties of \textit{RPM-3SAT}.
    
    \section{Solution framework} \label{sec:sol}
    
    For the purpose of our reduction, we will define building blocks for the puzzle, called gadgets, inspired by \textit{Gadget Area Hardness Framework} \cite{adler2020tatamibari}. The puzzle will consist of gadgets and external areas.
    
    A \textbf{gadget} is a subpuzzle, a connected subset of cells, consisting of mandatory and one or more optional areas. Each optional area is the place where two adjacent gadgets overlap, and cannot be adjacent to any other gadget. Any area assigned to a gadget must satisfy the following:
    \begin{enumerate}
        \item It covers the gadget's mandatory area.
        \item It is a subset of the gadget's entire area.
        \item It either fully covers or fully uncovers any optional area.
    \end{enumerate}
    A gadget's solution is a local solution to a subpuzzle with an assigned area that satisfies the above mentioned properties, as well as the \textit{Double Choco} properties from \autoref{chap:1}.
    
    A gadget's \textbf{profile table} is the list of its solutions. A profile table is complete if it contains all of a gadget's solutions.
    
    An \textbf{external area} is a subpuzzle consisting only of a mandatory area. Each external area should not interact with any gadget, and the external area must be solvable locally (only using its own cells). Hence, the external areas do not effect the solution to our puzzle and their only utilization is to cover the holes in the grid structure of the puzzle in \autoref{chap:4}.

    Given a puzzle consisting of gadgets and external areas, a \textbf{puzzle profile assignment} defines an area assigned to each gadget such that each area is disjointed from all the others, and a union of the solutions covers entire area of all gadgets. A profile assignment is valid if all assigned gadgets' areas represent the solution of respective gadgets, i.e., the assigned area appears in the gadgets' profile table. Observe that if $G_1$ and $G_2$ are adjacent gadgets overlapping on optional area $A_1$, in any valid profile assignment, if an area assigned to $G_1$ covers (uncovers) $A_1$, than $G_2$ uncovers (covers) $A_1$.

    \chapter{Arbitrarily shaped Double Choco is NP-complete} \label{chap:3}
    An arbitrarily shaped \textit{Double Choco} is a relaxed version of the original puzzle in which the board can have arbitrary shaped holes in its grid structure, but its cells still need to be connected.
    In this chapter we prove it to be \textit{NP-complete} by reduction from the \textit{Bent RPM-3SAT}. 
    
    With this relaxation, our puzzle can consist only of connected gadgets, without external areas.  We will describe the wire, variable, and clause gadgets used to make our reduction in \autoref{sec:wire}, \autoref{sec:variable}, and \autoref{sec:clause}, respectively,
    as well as a various supporting gadgets such as the rotation and wire boundary gadgets.
    
    When arranging our gadgets we will follow \autoref{fig:bent} and scale it up in some cases. Each wire starts from the wire boundary gadget that determines its value. The wire is used to transfer the value and is then connected to a variable gadget. The variable gadget connects the neighbouring wires so that all upward wires have the same value, all downward wires have the same value, while the upward and downward wires have the opposite value. After exiting the variable gadget, unused wires will be terminated in another wire boundary gadget, while the used wires will transfer their value to the clause gadget, potentially going through a rotation gadget which will rotate them $90\si{\degree}$ in the direction of the clause. 
    
    The wire boundary gadget has only one optional area which will either be fully covered or fully uncovered by its solution. When starting the wire, if the optional area is covered by the wire boundary gadget solution, it will assign the value \textit{true} to the wire, otherwise, if it is uncovered, it will assign \textit{false}. While all the other gadgets, except the clause gadget, will have paired optional areas of the same size, on one area they will receive the truth value while on the other they will output it. Exactly one of the paired areas will be covered by the gadget solution, in that way transferring the truth value. The clause gadget will have three optional areas, representing inputs which receives the truth values from incoming wires. It will enforce that at least one wire communicates \textit{true} because it can be solved only in that case.
    
    When a puzzle is created in the way described above, it will have each optional area covered completely by exactly one gadget. Hence, the gadget solutions that partially cover the optional area will never happen in any solution of the puzzle, even if the gadget is solvable in that case. Thus, we will not present them in the gadget's profile table. We will call the profile table that contains all the solutions except those the \textbf{\textit{sufficient profile table}}. Consequently, we have relaxed the third property of the gadget assigned area without changing the puzzle profile assignment.
    
    \newpage

        \begin{figure*}[htb]
            \centering
            \begin{subfigure}[t]{0.4\textwidth}
                \centering
                \includegraphics[scale = 0.35]{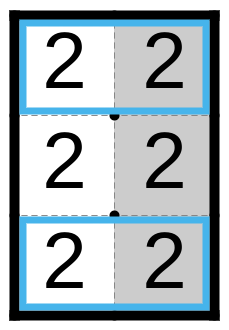}
                \caption[Network2]%
                {{\small An unsolved wire gadget. Optional areas are inside the blue rectangles, while the area outside them is mandatory.}}    
                \label{fig:uwire}
            \end{subfigure}
            \hfill
            \begin{subfigure}[t]{0.25\textwidth}  
                \centering 
                \includegraphics[scale = 0.35]{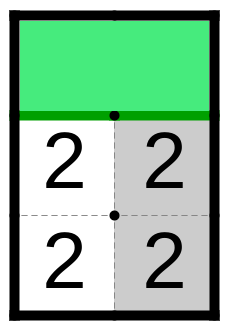}
                \caption[]%
                {{\small Wire communicating \textit{false}.}}    
                \label{fig:fwire}
            \end{subfigure}
            \hfill
            \begin{subfigure}[t]{0.25\textwidth}   
                \centering 
                \includegraphics[scale = 0.35]{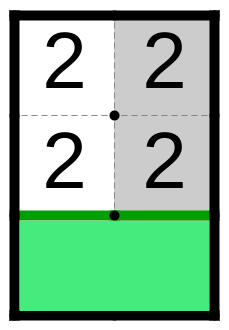}
                \caption[]%
                {{\small Wire communicating \textit{true}.}}    
                \label{fig:twire}
            \end{subfigure}
            \caption[]
            {\small  Wire gadget with the profile table. In figures (b) and (c) green area is not covered by the gadget solution. It is covered by an adjacent gadget, which determines the value of the wire. }
            \label{fig:wire}
        \end{figure*}
        \section{Wire gadget and supporting gadgets} \label{sec:wire}
        In \autoref{fig:wire} we present the wire gadget consisting of the odd number of rows, in this case three, and two columns, with an integer $\bm{2}$  in each cell, while cells in left and right column are colored differently. When determining the value in the subfigures (b) and (c), we assumed that the lower optional area receives the truth value while the upper optional area outputs it. In the opposite case, the values will be opposite.
        
        \begin{lemma} \label{lem:wire}
        The bottom optional area of the wire gadget is covered by the gadget assigned area if, and only if, the top optional area is not covered.
        \end{lemma}
        \begin{proof}
            By using the puzzle properties, if the bottom optional area is covered, it must be in the same block as the mandatory area. Since the top optional area has cells with an integer of $\bm{2}$, it therefore must be in a block of four cells. Thus, the top optional area cannot be covered by that gadget solution. Similarly, if the top optional area is covered, the bottom optional area remains uncovered. 
        \end{proof}
        
        \begin{lemma}
            The wire gadget's profile table in \autoref{fig:wire} is \textit{sufficient}.
        \end{lemma}
        \begin{proof}
            As a consequence of Lemma \autoref{lem:wire}, regardless of whether the bottom optional area is entirely uncovered or entirely covered by the solution, the solution should tile exactly four cells. Since those four cells contain integer $\bm{2}$, in both cases, there is exactly one solution, shown in \autoref{fig:fwire} and \autoref{fig:twire}, taking all the cells into one block.
        \end{proof}
        
        The wire can be extended to an arbitrary height by appending wire gadgets on top of each other, since the extended gadget will preserve its properties. Each appended wire gadget will extend the wire height by two.
        Observe in \autoref{fig:wire}, if we negate the colors, i.e. the cells in the left column are gray while the cells in the right column are white, the gadget will preserve its properties. Furthermore, the gadget is invariable with respect to rotation and mirroring. All gadgets in this study will satisfy these properties, since they directly follow from the puzzle properties. Therefore, in the rest of the thesis, we will use them without explicitly addressing them.

        \newpage
        \begin{figure*}[htb]
            \centering
            \begin{subfigure}[t]{0.35\textwidth}
                \centering
                \includegraphics[scale = 0.35]{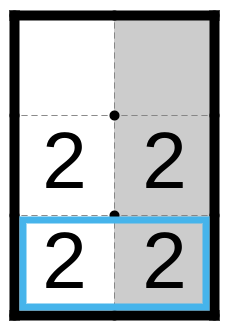}
                \caption[Network2]%
                {{\small An unsolved wire boundary gadget. The optional area is in the blue rectangles, while the area outside it is mandatory.}}    
                \label{fig:uending}
            \end{subfigure}
            \hfill
            \begin{subfigure}[t]{0.3\textwidth}  
                \centering 
                \includegraphics[scale = 0.35]{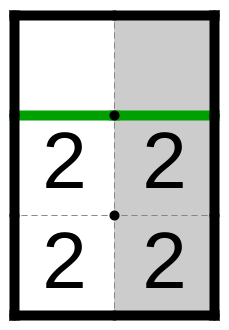}
                \caption[]%
                {{\small Starting wire communicating \textit{true}, or terminating wire communicating \textit{false}.}}    
                \label{fig:fending}
            \end{subfigure}
            \hfill
            \begin{subfigure}[t]{0.3\textwidth}   
                \centering 
                \includegraphics[scale = 0.35]{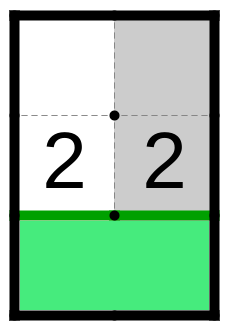}
                \caption[]%
                {{\small Starting wire communicating \textit{false}, or terminating wire communicating \textit{true}.}}    
                \label{fig:tending}
            \end{subfigure}
            \caption[]
            {\small  Wire boundary gadget with the profile table. The green area is not covered by the gadget solution.}
            \label{fig:ending}
        \end{figure*}
        In \autoref{fig:ending} we present the wire boundary gadget. This gadget will be used to start the new wires and to terminate unused ones. Observe that its value will be interpreted differently depending on whether it is used to start or terminate the wire. 
        
        \begin{lemma}
        The wire boundary gadget does not enforce wire value.
        \end{lemma}
        \begin{proof}
        Solutions for the both possible wire values are shown in \autoref{fig:fending} and \autoref{fig:tending}.
        \end{proof}
        
        \begin{theorem}
            The profile table of the wire boundary gadget is \textit{complete}.
        \end{theorem}
        \begin{proof}
        There are three possible cases in which solution, not shown in \autoref{fig:uending}, may exist. We will argue each case separately:
        \begin{enumerate}
            \item There exists another solution that covers the whole gadget and it is different than the one in \autoref{fig:fending}. Since our gadget has six cells, four of which have an integer $\bm {2}$, all four of them must be in the same block. This forces the two cells on the top of the  gadget to be tiled together. Hence, we reached the contradiction, because that solution is already shown in the figure.
            \item There exists another solution that covers the mandatory area, but it does not cover an optional area and it is different than the one in \autoref{fig:tending}. That solution should tile four cells, while having cells with an integer $\bm {2}$. The only way to tile them is with all four cells inside the one block. Again, we reach the contradiction since that solution is already shown.
            \item There exists a solution that covers the mandatory area and the one cell of the optional area. That solution should cover an odd number of cells, thus violating the property of the puzzle that white and gray areas should have the same size in each block. Hence, such a solution is impossible.
        \end{enumerate}
        \end{proof}
        
        Together we proved that the wire boundary gadget allows both possible wire values,  and it does not allow for any other solution. When it starts the wire, it will determine its value. When it terminates the wire, it will do so irrespective of its value.
        
        \newpage
        \begin{figure*}[htb]
            \centering
            \begin{subfigure}[t]{0.4\textwidth}
                \centering
                \includegraphics[scale = 0.25]{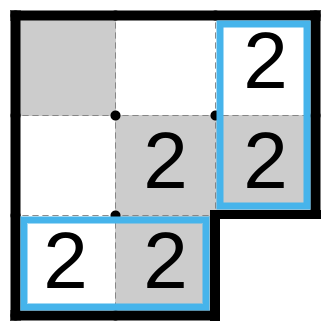}
                \caption[Network2]%
                {{\small An unsolved right rotation gadget. Paired optional areas are in blue rectangles, area outside them is mandatory.}}    
                \label{fig:urangle}
            \end{subfigure}
            \hfill
            \begin{subfigure}[t]{0.25\textwidth}  
                \centering 
                \includegraphics[scale = 0.25]{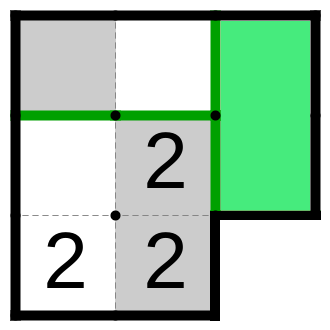}
                \caption[]%
                {{\small Redirecting wire communicating \textit{false} on the right.}}    
                \label{fig:frangle}
            \end{subfigure}
            \hfill
            \begin{subfigure}[t]{0.25\textwidth}   
                \centering 
                \includegraphics[scale = 0.25]{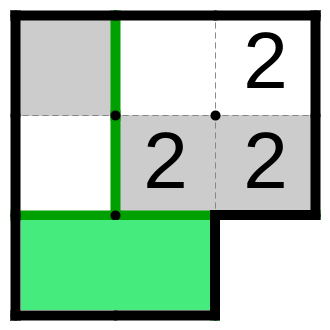}
                \caption[]%
                {{\small Redirecting wire communicating \textit{true} on the right.}}    
                \label{fig:trnagle}
            \end{subfigure}
            \vskip\baselineskip
            \begin{subfigure}[t]{0.4\textwidth}
                \centering
                \includegraphics[scale = 0.25]{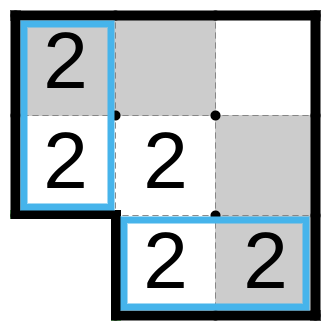}
                \caption[Network2]%
                {{\small An unsolved left rotation gadget. Paired optional areas are in blue rectangles, area outside them is mandatory.}}    
                \label{fig:ulangle}
            \end{subfigure}
            \hfill
            \begin{subfigure}[t]{0.25\textwidth}  
                \centering 
                \includegraphics[scale = 0.25]{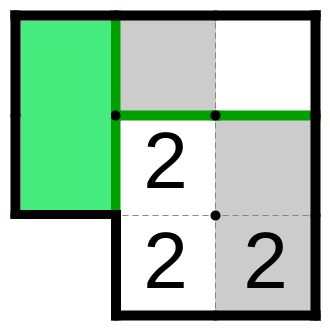}
                \caption[]%
                {{\small Redirecting wire communicating \textit{false} on the left.}}    
                \label{fig:flangle}
            \end{subfigure}
            \hfill
            \begin{subfigure}[t]{0.25\textwidth}   
                \centering 
                \includegraphics[scale = 0.25]{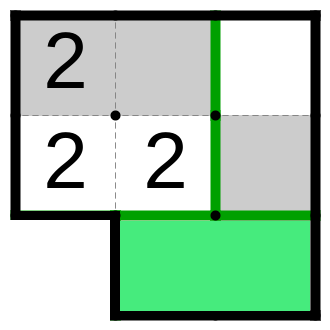}
                \caption[]%
                {{\small Redirecting wire communicating \textit{true} on the left.}}    
                \label{fig:tlangle}
            \end{subfigure}
            \caption[]
            {\small  Rotation gadgets with their profile table. Green areas are  not covered by the gadget solution. }
            \label{fig:angle}
        \end{figure*}
        
        We have presented the right and the left rotation gadgets in \autoref{fig:angle}. The values shown apply in the case when the lower optional area receives the truth value while the other optional area outputs it. In the opposite case, the values will be opposite. Observe that they are the same gadget, because we can get the left rotation gadget from the right rotation gadget after applying vertical mirroring and color inversion. Hence, it is sufficient to discuss only the right rotation gadget.
        
        \begin{lemma} \label{lem:rot}
            The right optional area of the right rotation gadget is covered if, and only if, the bottom optional area is uncovered.
        \end{lemma}
        \begin{proof}
            If the bottom optional area is covered by the gadget solution, then it must be in the same block as adjacent cells in the mandatory area, in order to have a connected gray and white area. After that, the top two cells of the mandatory area are forced to be in the same block, leaving the right optional area uncovered. If the right optional area is covered by the gadget solution, then, by similar argument, the bottom optional area remains uncovered.
        \end{proof}
        
        \begin{lemma}
            The right rotation gadget's profile table is sufficient.
        \end{lemma}
        \begin{proof}
            We have already proven this statement in the proof of Lemma \autoref{lem:rot}, by arguing that there exists exactly one solution in which the bottom area is covered and one in which it is uncovered, which are shown in \autoref{fig:frangle} and \autoref{fig:trnagle}, respectively.
        \end{proof}
        
        The rotation gadgets will be used at most once per wire, and when the left and the right wires of the clause need to rotate toward the clause, as shown in \autoref{fig:bent}.

        \newpage

        \begin{figure}[htp]
        {\hfill}
        \centering
        \begin{minipage}[t]{.65\textwidth}
          \centering
          \includegraphics[width=\linewidth]{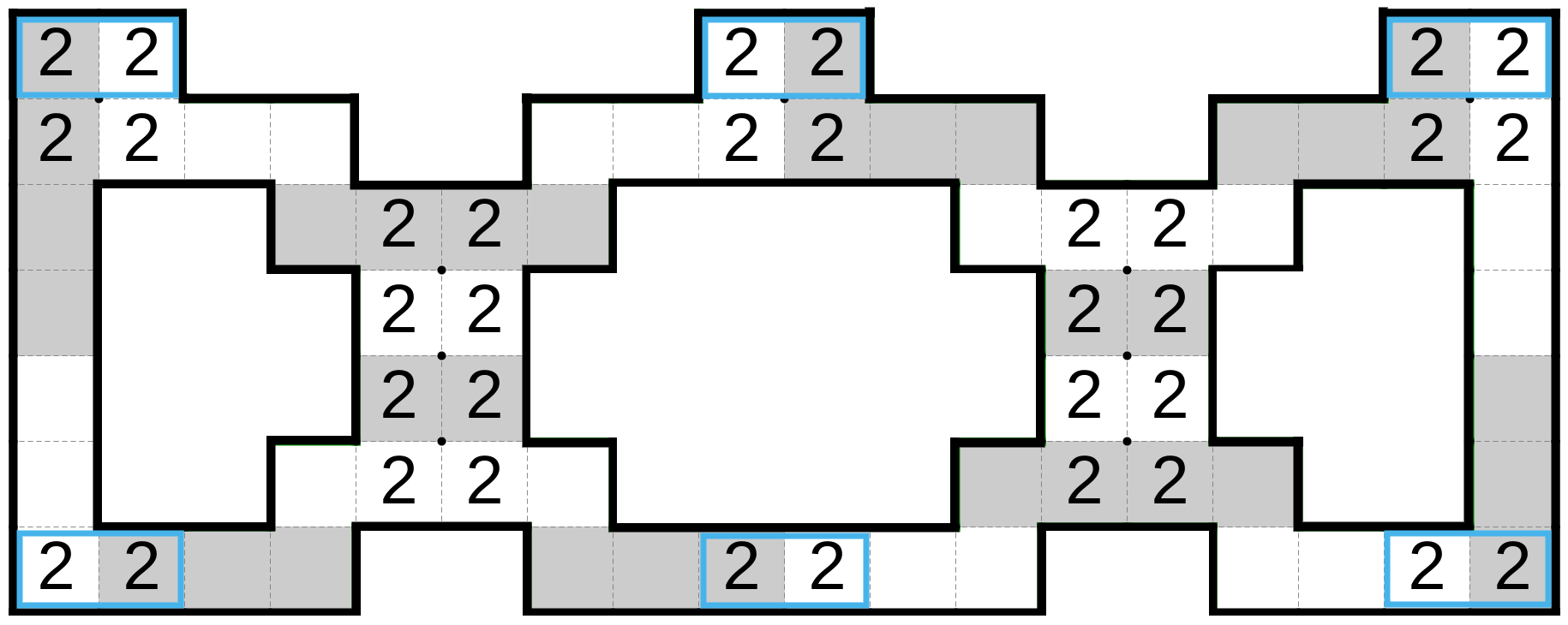}
          \captionof{figure}{The variable gadget. Optional areas are in the blue rectangles and they represent the places where the wires enter and exit gadget.}
          \label{fig:var}
        \end{minipage}%
        {\hfill}
        \centering
        \begin{minipage}[t]{.3\textwidth}
          \centering
          \includegraphics[width=.74\textwidth]{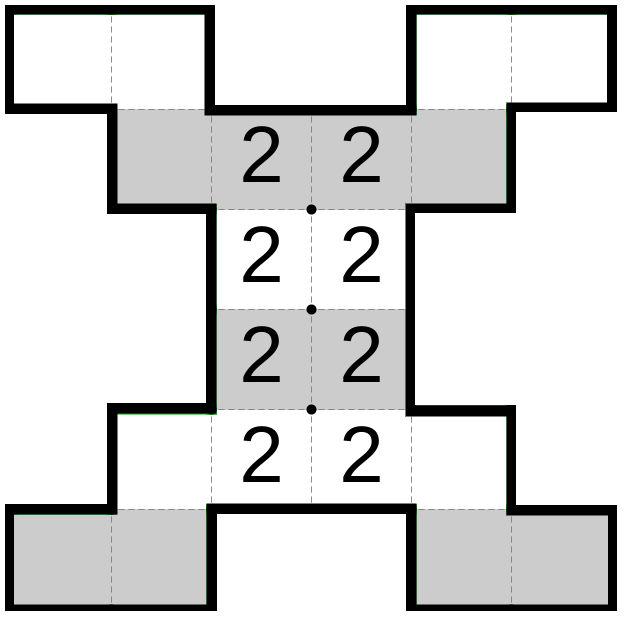}
          \captionof{figure}{Equalizer}
          \label{fig:test2}
        \end{minipage}
        {\hfill}
        \end{figure}
        \section{The variable gadget} \label{sec:variable}

        In \autoref{fig:var} we have presented the variable gadget. It consists of a sequence of \textit{equalizers} and two vertical segments. Equalizers connect neighbouring wires while vertical segments are used to transfer the value of two wires with only one neighbour, i.e. the furthest left and furthest right wires. 
        
        The first gadget the wire enters after starting in the wire boundary gadget will be the variable gadget. Therefore, we know that optional areas where the wires enter the gadget will either be fully covered of full uncovered. We will prove all claims in this section under the constraint that optional areas are never partially taken at the place the wire enters the gadget.
        
        Observe that optional areas are paired in such a way that for each bottom optional area there is a top optional area in the same column, and they represent places where wires enter and exit the gadget. Wires can enter on either side, depending on whether we need upward or downward wire. In the case we need upward wire, the wire will enter the gadget on the bottom side and will exit on the top optional area above it.
        
        The equalizer connects the neighbouring wires so that either both upper optional areas or both lower optional areas are forced to be covered by the gadget's solution. Because the wires are connected to those optional areas, if the wires are going in the same direction, they will have the same value. If one is going upward and other is going downward, they will have the opposite values. The equalizer consists of:
        
        \begin{itemize}
            \item Two columns, each comprising of four cells, are collectively called a rectangle. Cells of the rectangle have an integer \textbf{2}. A rectangle acts as an equality constraint. 
            \item Four L shaped regions which will enforce the tiling of the optional areas based on the tiling of the rectangle.
        \end{itemize}  

        \begin{lemma} \label{lem:eqvi}
            Equalizer has only two valid ways to cover its rectangle.
        \end{lemma}
        \begin{proof}
            Let's order the rows of the equalizer's rectangle from top to bottom, i.e. the top row is \textit{row 1} and the bottom row is \textit{row 4}. The second row can be tiled together either with the first or with the third row.
            
            If the second row is in the same block as the third row (see \autoref{fig:tvar}), then the top left cell of the rectangle is forced to be tiled with the top left L shaped segment. A similar pattern applies on the top right, bottom left, and bottom right cells.
            
            Otherwise, if the second row is in the same block as the first row (see \autoref{fig:fvar}), then the third row must be in the same block as the forth row. Thus, forcing four cells adjacent to the rectangle to be tiled with their only uncovered neighbours. This leaves the two the furthest left and furthest right cells of the equalizer uncovered, which later must be tiled with the rest of the gadget cells.
        \end{proof}
        \begin{figure*}[t]
            \centering
            \begin{subfigure}[t]{0.47\textwidth}
                \centering
                \includegraphics[width=\textwidth]{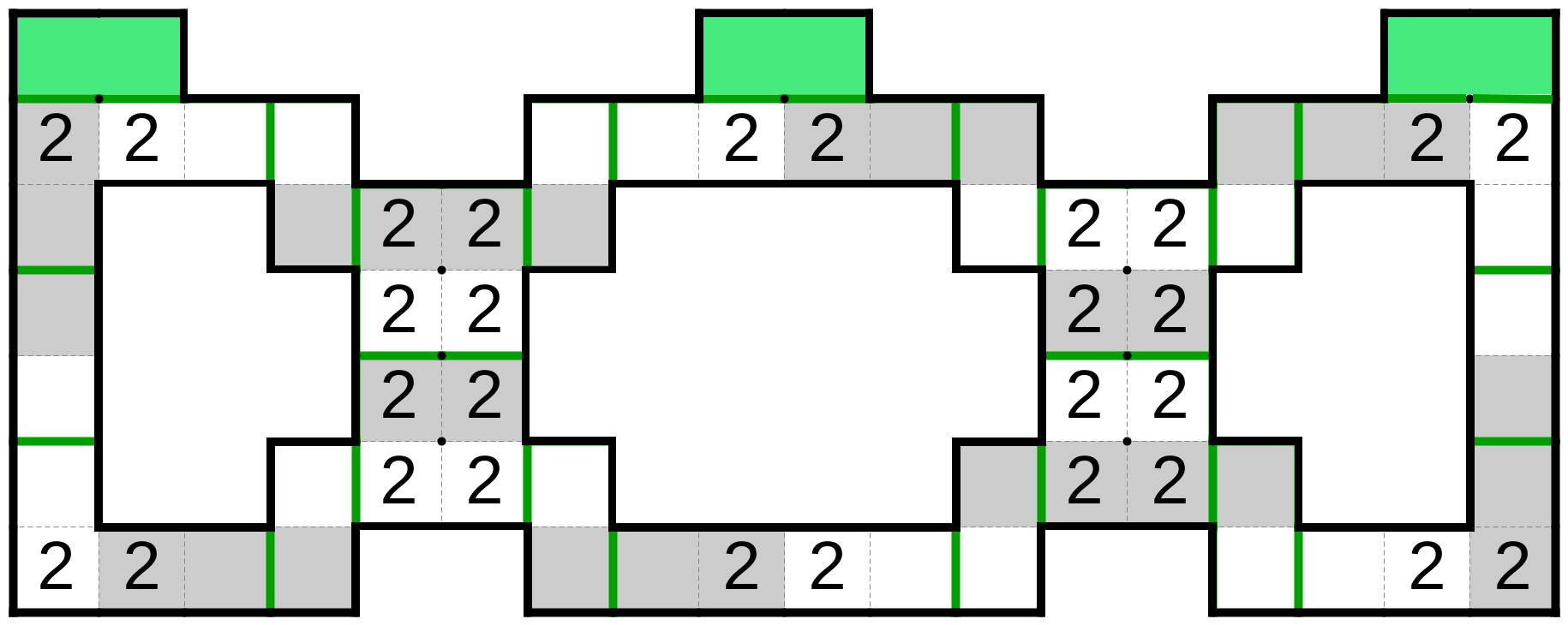}
                \caption[Network2]%
                {{\small Value \textit{false}.}}    
                \label{fig:fvar}
            \end{subfigure}
            \hfill
            \begin{subfigure}[t]{0.47\textwidth}  
                \centering 
                \includegraphics[width=\textwidth]{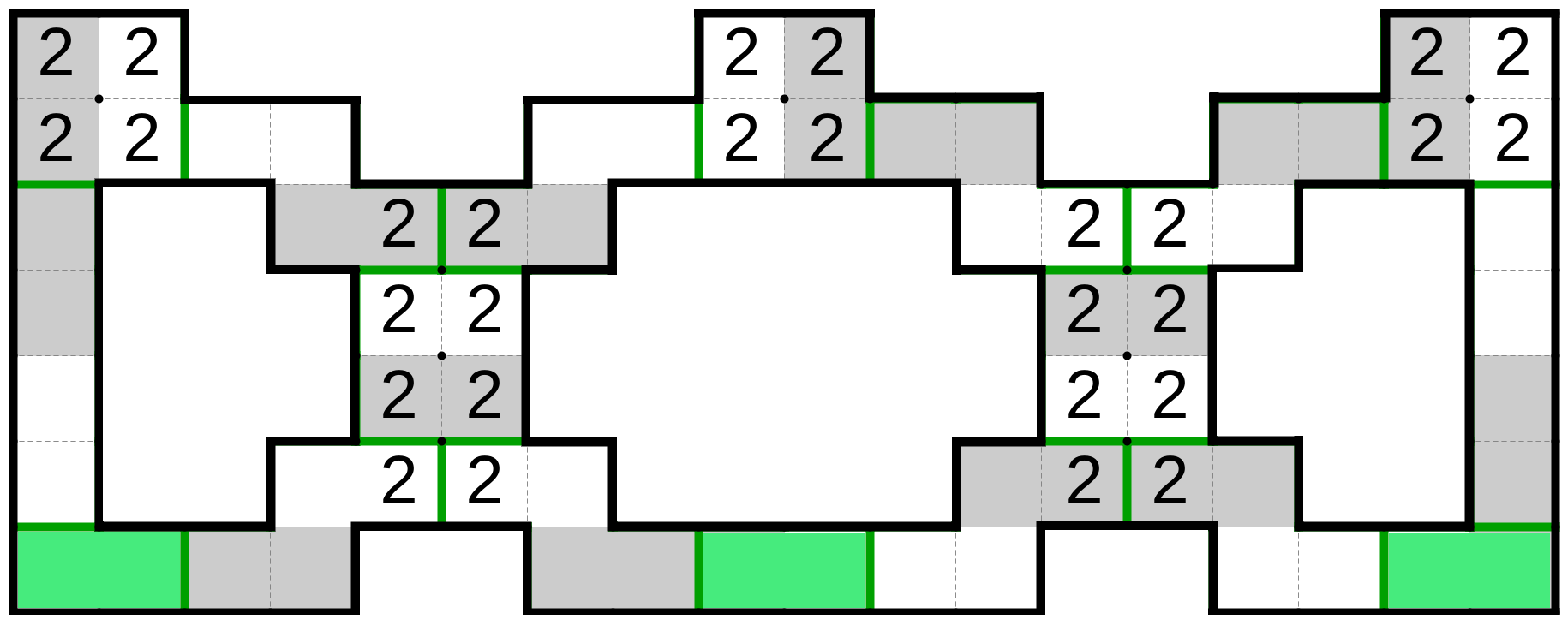}
                \caption[]%
                {{\small Value \textit{true}.}}    
                \label{fig:tvar}
            \end{subfigure}
            \caption[]
            {\small  The variable gadget with its profile table profile table.}
            \label{fig:varr}
        \end{figure*}
        \begin{claim} \label{claim:1}
            Each vertical segment has exactly two valid coverings. 
        \end{claim}
        \begin{proof}
            Coverings are shown in \autoref{fig:fvar} and \autoref{fig:tvar}. It is easily observable that those are the only two coverings.
        \end{proof}

        \begin{theorem} \label{thm:var}
        The variable gadget has a valid solution if, and only if, all top optional areas have the same tiling, all the bottom optional areas have the same tiling, and top and bottom optional areas have opposite tiling.
        \end{theorem}
        \begin{proof}
            We will split the proof of this theorem into the next two lemmas.
        \end{proof}
        
        \begin{lemma} \label{lem:opp_areas}
            The bottom optional area is covered by the gadget solution if, and only if, its corresponding top optional area is not covered by that solution.
        \end{lemma}
        \begin{proof}
            Suppose the wire enters the gadget through one of the bottom optional areas. If the wire communicates a \textit{true} value, it will force its adjacent equalizers to be tiled as shown in \autoref{fig:tvar}. If that wire is in between one of the two equalizers, the two cells of the mandatory area above that optional area will be forced to tile with the top optional area. The same argument follows if the wire is connected to the equalizer and the vertical segment, since the vertical segment will be forced to tile in a size four block. 
            
            Otherwise, if that wire communicates \textit{false}, the bottom optional area must be in the same block with its only two adjacent cells. This will force its adjacent equalizers, or, if the wire is located on the far left or the far right, the equalizer and the vertical segment, to be tiled as shown in \autoref{fig:fvar}. Thus, forcing the four cells of the mandatory area, above the bottom optional area, to be tiled together, leaving the top optional area uncovered.
            
            Observe that with color inversion and without the top optional areas, our gadget is symmetric about the horizontal line passing through the middle of the equalizers. Also, observe that the top optional area can be tiled only with the two adjacent cells of the mandatory area. Hence, when the wire enters the gadget through one of the top optional areas, a similar argument as in the previous case will apply, with inverted cases. That is, when the wire enters on top optional area with the value \textit{true}, we will argue in a similar way as when the wire enters on the bottom with the \textit{false} value, and conversely.
        \end{proof}

        \begin{lemma}
            The bottom optional area of the gadget is covered by the gadget solution, if, and only if, the bottom optional area next to it is covered as well and likewise for the top optional areas.
        \end{lemma}
        \begin{proof}
            As a direct consequence of the previous lemma, it is enough to prove this statement only for the bottom optional area.
            
            Suppose a wire enters the gadget through one of the bottom optional areas with the value \textit{false}, and the wire next to it enters the gadget through the bottom with the value \textit{true}, or equivalently on the top with the value \textit{false} as consequence of Lemma \autoref{lem:opp_areas}. The former wire will force the equalizer to be tiled as in \autoref{fig:fvar}, leaving the cell next to the bottom optional area of the latter wire uncovered, and unable to be tiled.
            
            For the other direction, suppose the wire enters the gadget through one of the bottom optional areas with the value \textit{true}, and the wire next to it enters the gadget through the top optional area with the value \textit{true}, or equivalently on the bottom with the value \textit{false}. The former wire will force the equalizer to be tiled as in \autoref{fig:tvar}. The bottom optional area of the latter wire must be in the same block as the two cells adjacent to it, but that can be done since one of them is already covered by the tiling of the equalizer.
        \end{proof}

        By proving \autoref{thm:var} we have confirmed that the variable gadget will have the desired properties, i.e. that all the upward wires will have the same value, all the downward wires will have the same value, and upward and downward wires will have the opposite value. Observe that the variable gadget can be extended from either side to fit an arbitrary number of wires, by adding more equalizers (see \autoref{fig:variable}).
        
        \begin{lemma}
            The profile table of the variable gadget is \textit{sufficient}.
        \end{lemma}
        \begin{proof}
            The statement directly follows from \autoref{thm:var}, Lemma \autoref{lem:eqvi}, and Claim \autoref{claim:1}.
        \end{proof}
        
      The variable gadget is the first gadget into which the wire enters after starting in the wire boundary gadget. Therefore, we will assume that each wire will start from the boundary wire gadget which is directly connected to the variable gadget. The wire will come out on the other side, if that wire is used, or there will be another wire boundary gadget to terminate unused wire. That we create the \textit{complete variable gadget} (see \autoref{fig:variable}).
        
        \begin{theorem} \label{thm:comp-var}
            Every upward wire of the complete variable gadget will communicate the same value, every downward wire will communicate the same value, and upward and downward wires will communicate the opposite values.
        \end{theorem}
        \begin{proof}
            The proof follows from the construction of the complete variable gadget and \autoref{thm:var}.
        \end{proof}
        
        In \autoref{fig:variable}, we show the construction of the complete variable gadget, with the two upward wires (the leftmost and the rightmost wires), two downward wires (the second and the third wire), and one terminated wire. Upward wires communicate the \textit{false} value, while the downward wires communicate the \textit{true} value, which is as expected since the complete variable gadget has a \textit{false} value.
        
        \begin{figure*}[t]
            \centering
            \includegraphics[width=\textwidth]{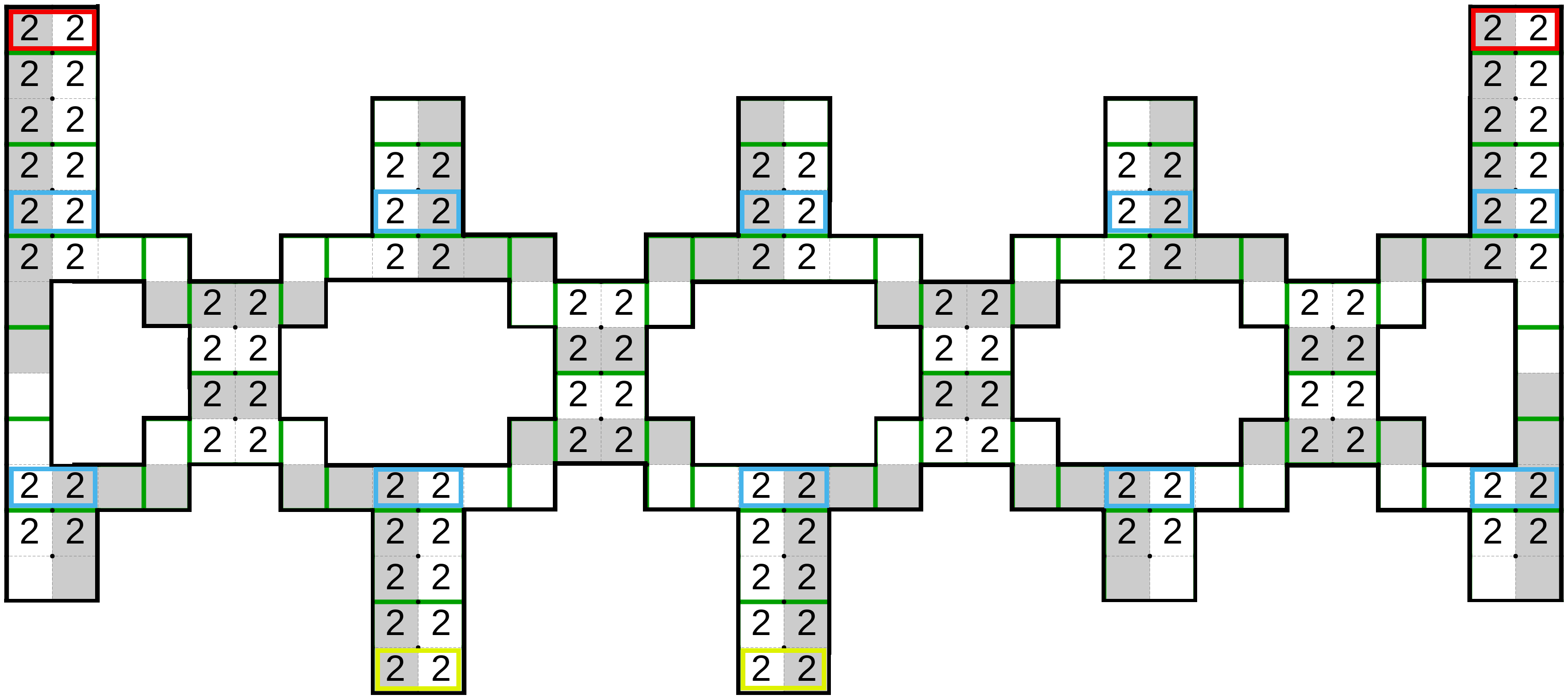}
            \caption[]
            {\small  The complete variable gadget with the \textit{false} value. Intersections of optional areas are in blue rectangles, optional areas covered by the solution of the gadget are in yellow rectangles, while optional areas not covered by the solution are in red rectangles. }
            \label{fig:variable}
        \end{figure*}

        \section{Clause gadget} \label{sec:clause}
        \begin{figure*}[t]
            \centering
            \begin{subfigure}[b]{0.23\textwidth}
                \centering
                \includegraphics[width=\textwidth]{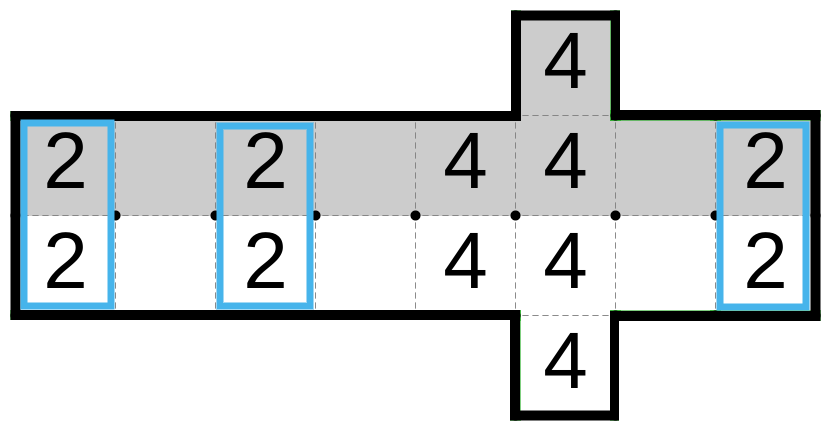}
                \caption[Network2]%
                {{\small Unsolved gadget. }}    
                \label{fig:3or-1}
            \end{subfigure}
            \hfill
            \begin{subfigure}[b]{0.23\textwidth}  
                \centering 
                \includegraphics[width=\textwidth]{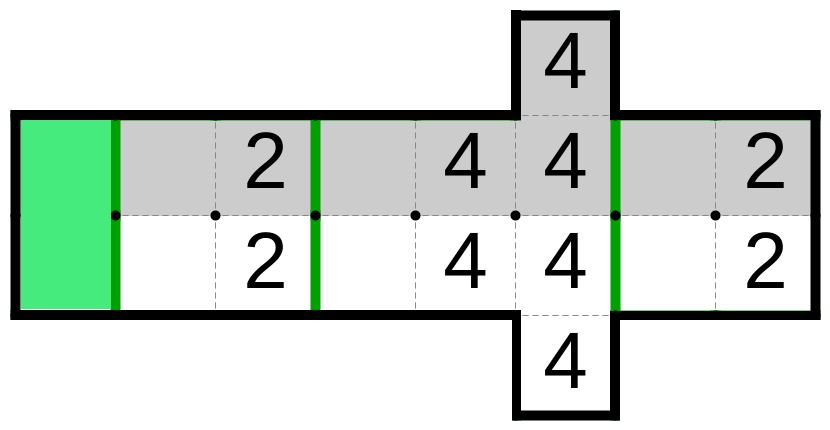}
                \caption[]%
                {{\small \textit{(true, false, false)}}}    
                \label{fig:mean and std of net24}
            \end{subfigure}
            \hfill
            \begin{subfigure}[b]{0.23\textwidth}   
                \centering 
                \includegraphics[width=\textwidth]{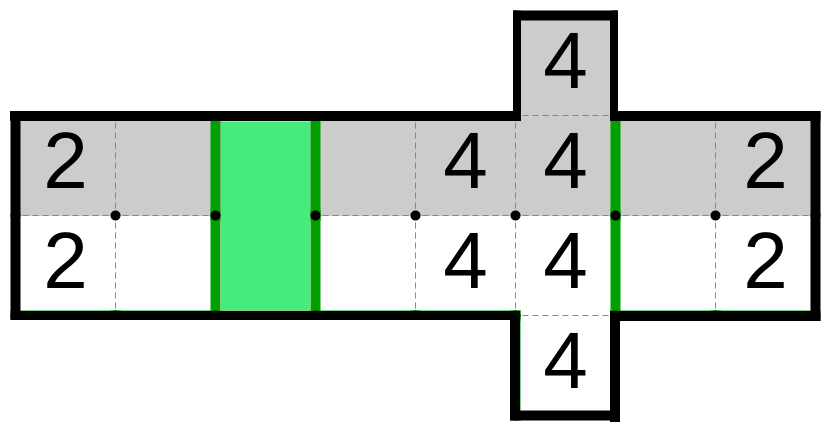}
                \caption[]%
                {{\small \textit{(false, true, false)}}}    
                \label{fig:mean and std of net44}
            \end{subfigure}
            \hfill
            \begin{subfigure}[b]{0.23\textwidth}
                \centering
                \includegraphics[width=\textwidth]{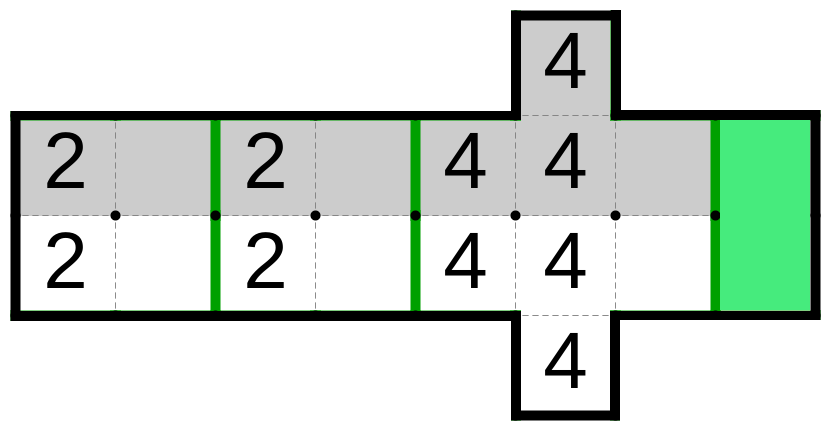}
                \caption[Network2]%
                {{\small \textit{(false, false, true)}}}    
                \label{fig:mean and std of net14}
            \end{subfigure}
            \vskip\baselineskip
            \begin{subfigure}[b]{0.23\textwidth}  
                \centering 
                \includegraphics[width=\textwidth]{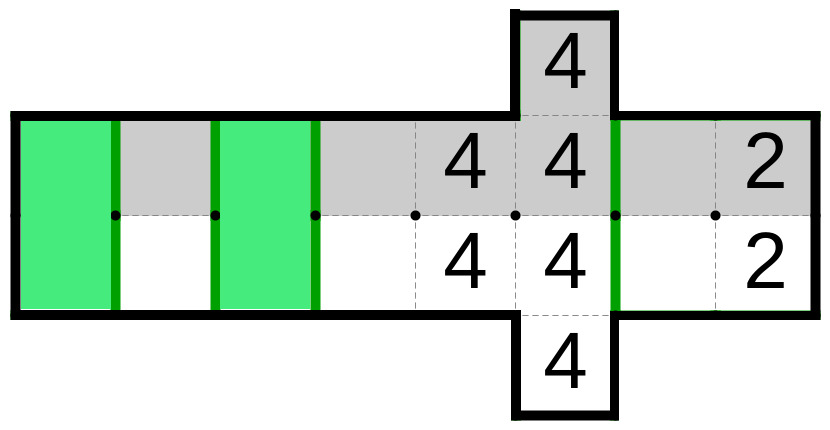}
                \caption[]%
                {{\small \textit{(true, true, false)}}}    
                \label{fig:mean and std of net24}
            \end{subfigure}
            \hfill
            \begin{subfigure}[b]{0.23\textwidth}   
                \centering 
                \includegraphics[width=\textwidth]{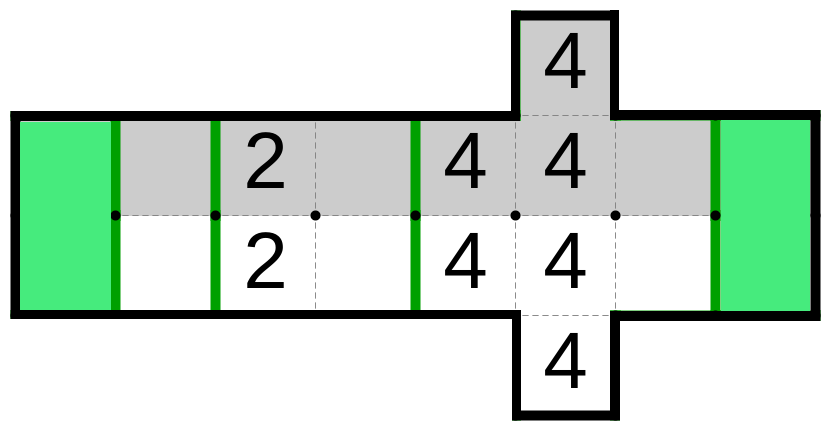}
                \caption[]%
                {{\small \textit{(true, false, true)}}}    
                \label{fig:mean and std of net44}
            \end{subfigure}
            \hfill
            \begin{subfigure}[b]{0.23\textwidth}
                \centering
                \includegraphics[width=\textwidth]{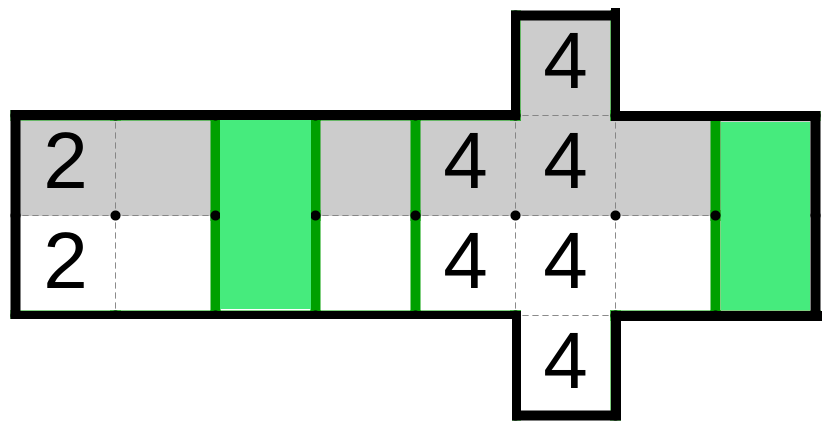}
                \caption[Network2]%
                {{\small \textit{(false, true, true)}}}    
                \label{fig:mean and std of net14}
            \end{subfigure}
            \hfill
            \begin{subfigure}[b]{0.23\textwidth}  
                \centering 
                \includegraphics[width=\textwidth]{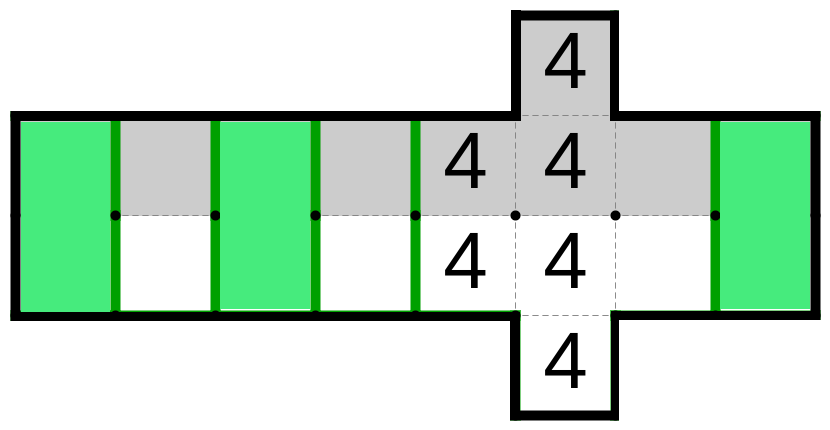}
                \caption[]%
                {{\small \textit{(true, true, true)}}}    
                \label{fig:mean and std of net24}
            \end{subfigure}
            \caption[ ]
            {\small The OR subgadget with the profile table. In the first figure, optional areas are in the blue rectangles. The other figures show possible solutions when at least one optional area is covered by an adjacent gadget. Solutions are shown as \textit{(left, middle, right)} based on whether the corresponding optional area is fully covered or fully uncovered by an adjacent gadget.} 
            \label{fig:OR}
        \end{figure*}
        
        In this section we will present the last required piece for our reduction, the \textit{clause gadget}. It should receive the value from the three wires, and enforce that at least one of the incoming wires communicates the \textit{true} value. However, as it turns out, the gadget is not as simple as it looks at first glance, hence we will construct it step by step, through the introduction of subgadgets, which will eventually combine into the clause gadget.
        
        After the wire exits the variable gadget, it either directly enters the clause gadget or goes through the rotation gadget. Hence, its optional areas are either fully covered or fully uncovered and we will assume so, and not argue the case when the optional areas are partially covered.
        
        In \autoref{fig:OR}, the \textit{OR subgadget} is shown. It consists of three optional areas separated by mandatory areas.
        
        \begin{lemma} \label{lem:OR}
            The OR subgadget is solvable if, and only if, at least one optional area is covered by its adjacent gadget. 
        \end{lemma}
        \begin{proof}
            For the latter direction, the proof is shown in \autoref{fig:OR} (b) - (h), where we show one gadget solution for each combination. Observe that there exists one more solution for the cases shown in (g) and (h), in which the segment of the cells with an integer \textbf{4} are in the same block as the two cells on their left, while the cells on their right are in a size two block.
            
            On the other hand, suppose that there exists a solution to the OR subgadget, which covers all three optional areas. The left optional area must be in the same square block as its two adjacent cells, forcing the middle optional area to be in the same square block as the two cells on its right. The right optional area can be tiled only with the cells on its left. Hence, leaving the six cells with an integer \textbf{4} to be uncovered. Since they are the only uncovered cells, they cannot be tiled. 
        \end{proof}
        \newpage
        Even though the OR subgadget has the desired property of the clause gadget, it cannot be the clause gadget since it has a problem. The problem which arises is that the middle wire cannot be connected to the OR subgadget without overlapping of their mandatory areas, which is not allowed. So we come up with an idea to resolve it:
        
        \begin{figure*}[htb]
            \centering
            \includegraphics[width=0.65\textwidth]{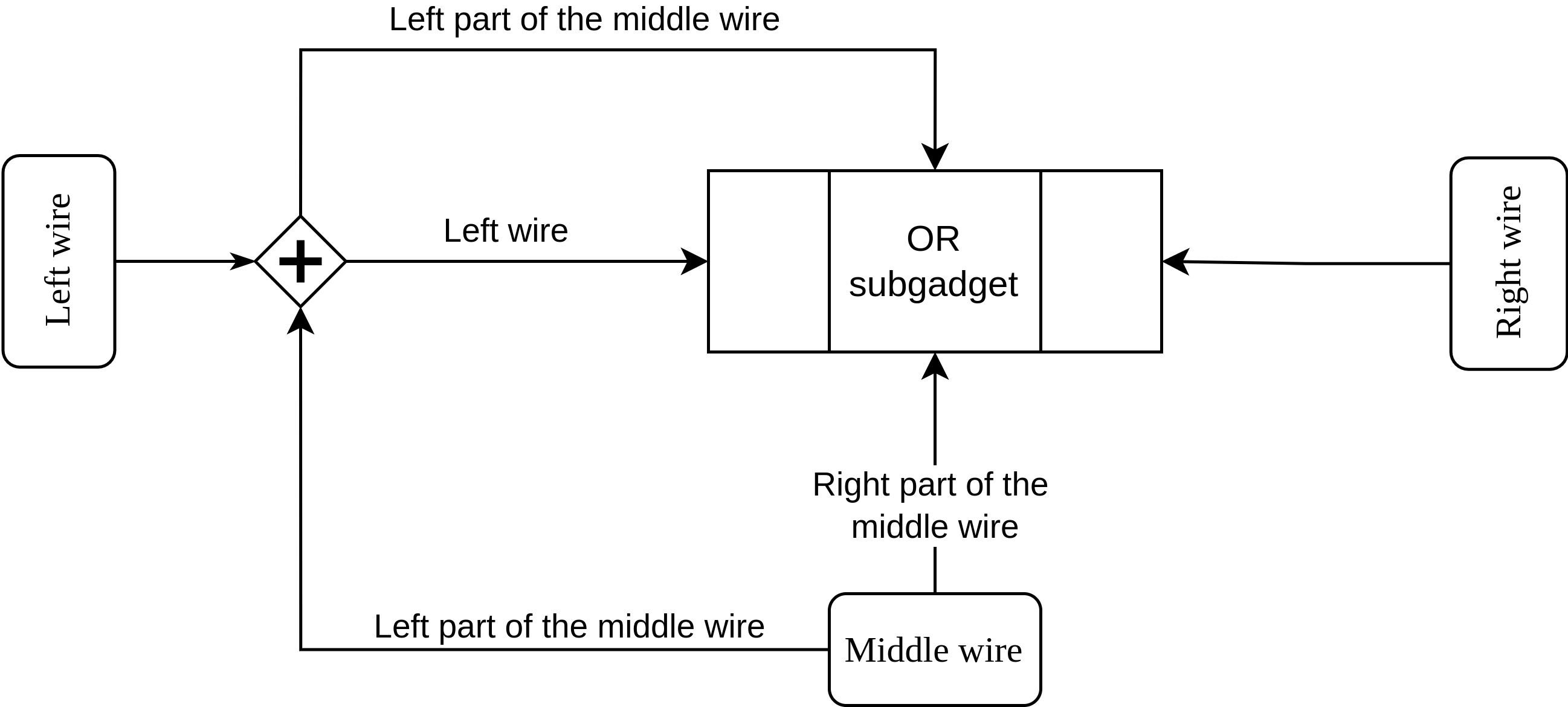}
            \caption[]
            {\small Idea for the solution of the OR subgadget problem. }
            \label{fig:clause-idea}
        \end{figure*}
        
        In \autoref{fig:clause-idea}, we presented the idea for the solution to the above mentioned problem by introducing another element represented by \rotatebox[origin=c]{45}{ $\squareop{\rotatebox[origin=c]{45}{+}}$}, which we will call the \textit{crossover subgadget}. We first split the middle wire by columns into its left and the right part. We then connected the right part of the middle wire with the lower cell of the middle optional area of the OR subgadget. The left part of the middle wire is navigated to the bottom of the crossover subgadget. Meanwhile, the left wire enters the crossover subgadget on its left side. The crossover subgadget should convey the value of the left wire on its right side, from where the wire will be connected to the left optional area of the OR subgadget. The left part of the middle wire should exit on the top of the crossover subgadget, from where it will go toward the OR subgadget and will be connected to the upper cell of the middle optional area.

        \begin{figure*}[htb]
            \centering
            \begin{subfigure}[t]{0.19\textwidth}
                \centering
                \includegraphics[scale = 0.12]{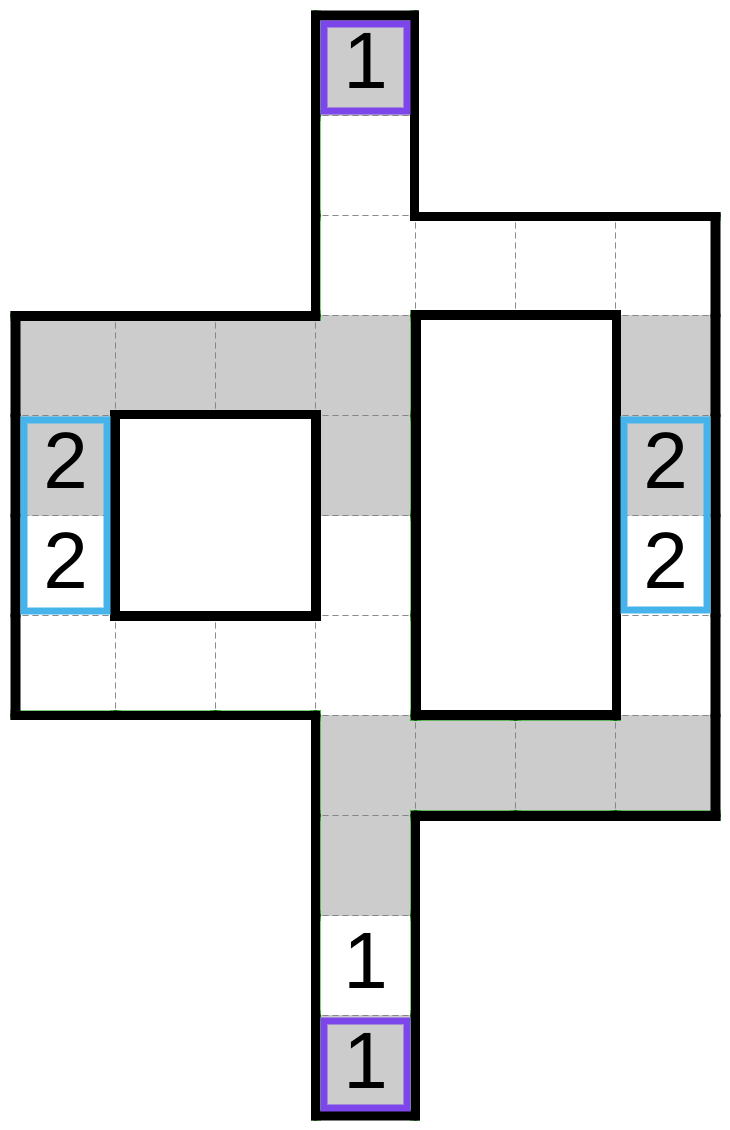}
                \caption[Network2]%
                {{\small An unsolved crossover subgadget. }}    
                \label{fig:uinv}
            \end{subfigure}
            \hfill
            \begin{subfigure}[t]{0.19\textwidth}   
                \centering 
                \includegraphics[scale = 0.12]{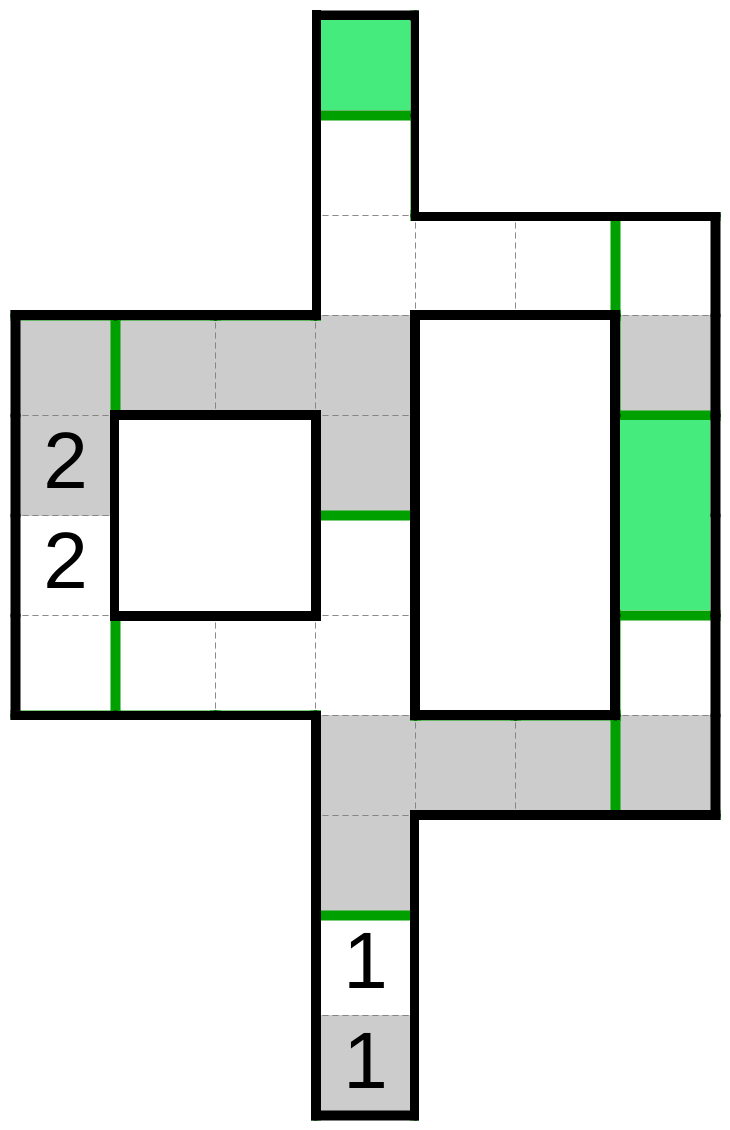}
                \caption[]%
                {{\small \textit{(false, false)}}}    
                \label{fig:tinv}
            \end{subfigure}
            \hfill
            \begin{subfigure}[t]{0.19\textwidth}  
                \centering 
                \includegraphics[scale = 0.12]{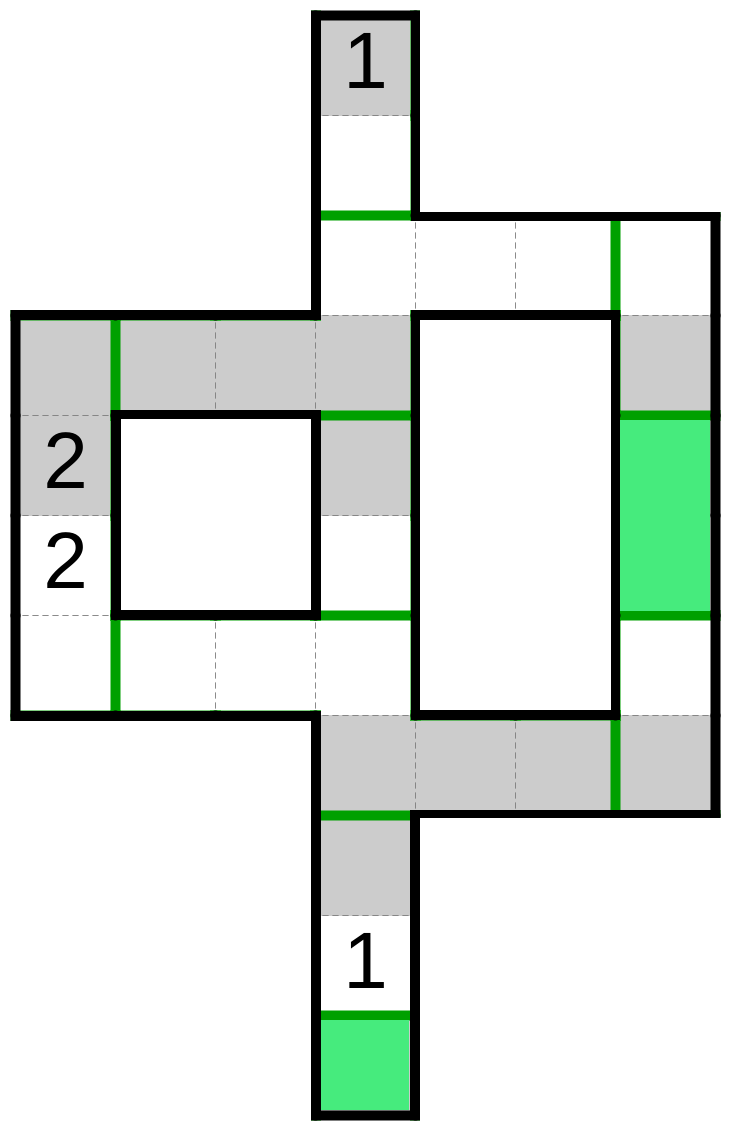}
                \caption[]%
                {{\small \textit{(false, true)}}}    
                \label{fig:finv-1}
            \end{subfigure}
            \hfill
            \begin{subfigure}[t]{0.19\textwidth}  
                \centering 
                \includegraphics[scale = 0.12]{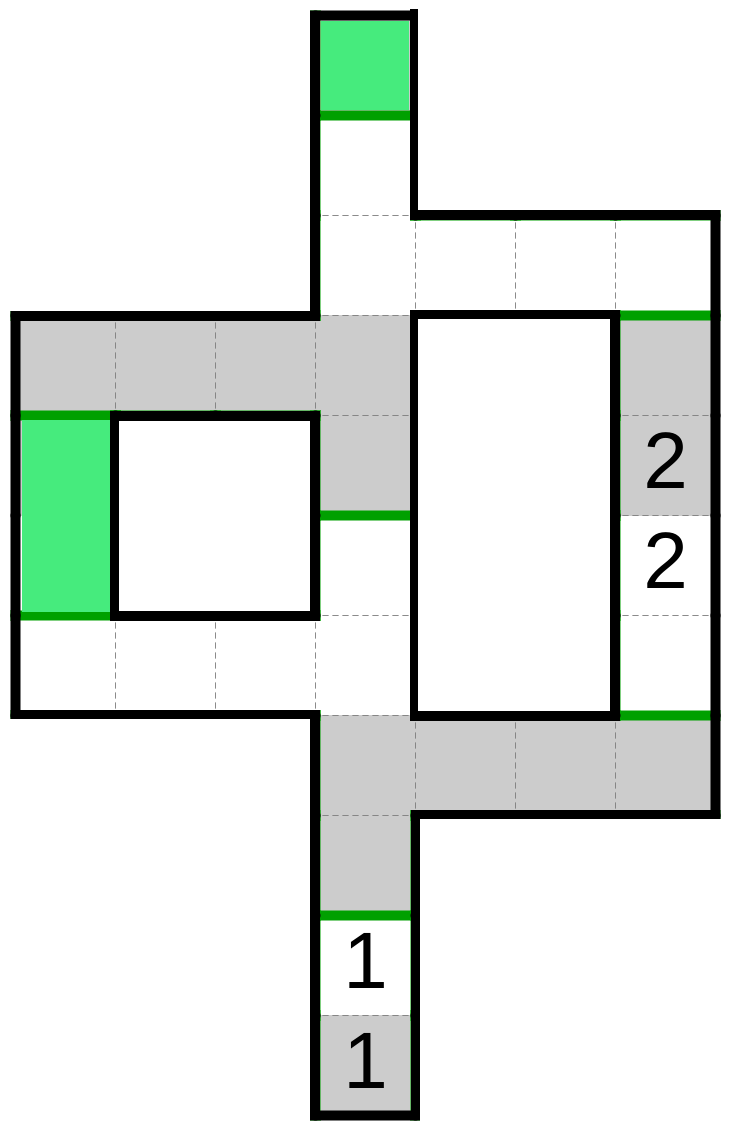}
                \caption[]%
                {{\small \textit{(true, false)}}}    
                \label{fig:finv-2}
            \end{subfigure}
            \hfill
            \begin{subfigure}[t]{0.19\textwidth}  
                \centering 
                \includegraphics[scale = 0.12]{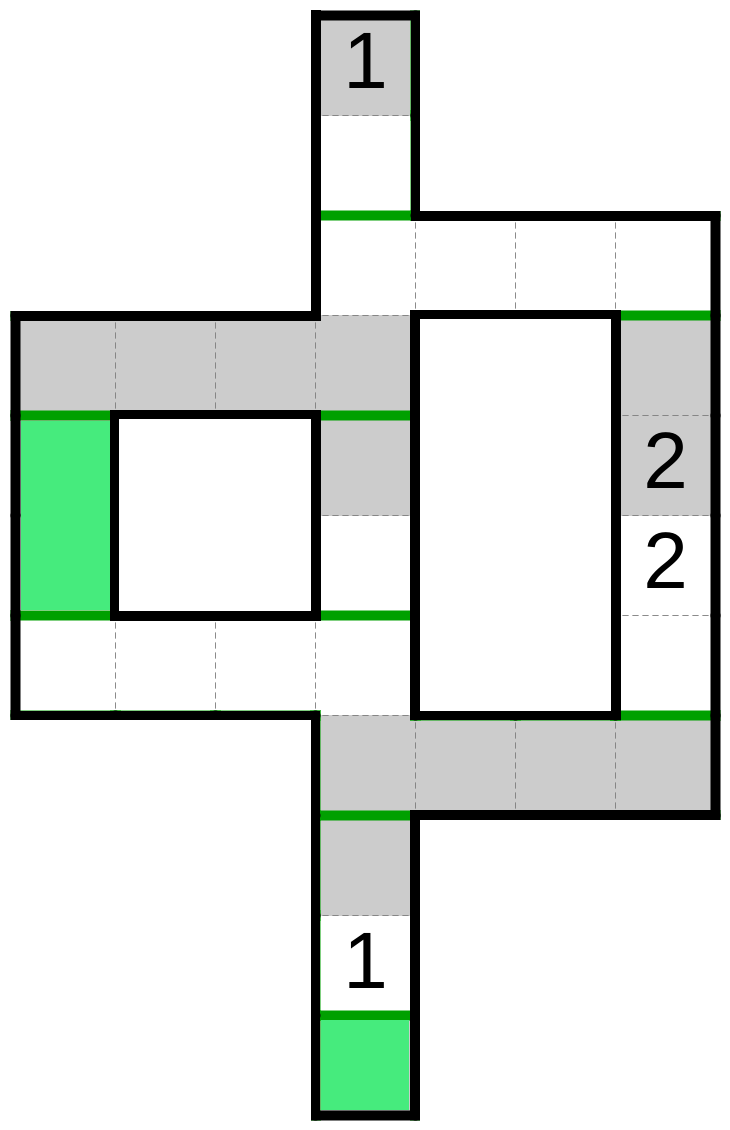}
                \caption[]%
                {{\small \textit{(true, true)}}}    
                \label{fig:finv-2}
            \end{subfigure}
            \caption[]
            {\small The crossover subgadget with its profile table. The paired optional areas are in the rectangle of the same color. Cells not covered by the subgadget solution are filled with green. In the subfigures (b) - (e), the solutions are shown as \textit{(value of left wire, value of the middle (its left part) wire)}.}
            \label{fig:cross}
        \end{figure*}
        
        The crossover subgadget is shown in \autoref{fig:cross}. The left wire will enter the subgadget on its left optional area, and it will exit it on the opposite side. The left part of the middle wire will be connected to the gadget on its bottom optional area, while exiting it on its top optional area.
        
        \begin{lemma} \label{lem:cross}
            The optional area of the crossover subgadget is fully covered by the gadget solution if, and only if, its pair is fully uncovered.
        \end{lemma}
        \begin{proof}
            In \autoref{fig:cross} (b) - (d), we show the solutions for each combination of the of the optional area coverings. It is straightforward to check that any other tiling of the subgadget is not prohibited.
        \end{proof}
        
        \begin{lemma}
            The profile table of the crossover subgadget is \textit{sufficient}.
        \end{lemma}
        \begin{proof}
            By the same argument as in Lemma \autoref{lem:cross}.
        \end{proof}
        
        Thus, we have proven that the crossover subgadget will convey the value of the incomming wires on the opposite sides from the ones they enter, as we wanted.
        
        We can now present the \textit{clause gadget}, constructed as the combination of the previous subgadgets, in \autoref{fig:clause}. In the figure we can observe that the crossover subgadget is on the left side, and its right optional area is glued together with the left optional area of the OR subgadget. However, there is still one part of the gadget we did not discuss: the part that connects the middle optional area of the clause gadget with the middle optional area of the OR subgadget. We will call it the \textit{middle wire split subgadget} (see \autoref{fig:wsplitter}). We will first discuss its properties then return to the clause gadget.
        
        \begin{figure*}[htb]
            \centering
            \begin{subfigure}[t]{0.55\textwidth}
                \centering
                \includegraphics[scale = 0.15]{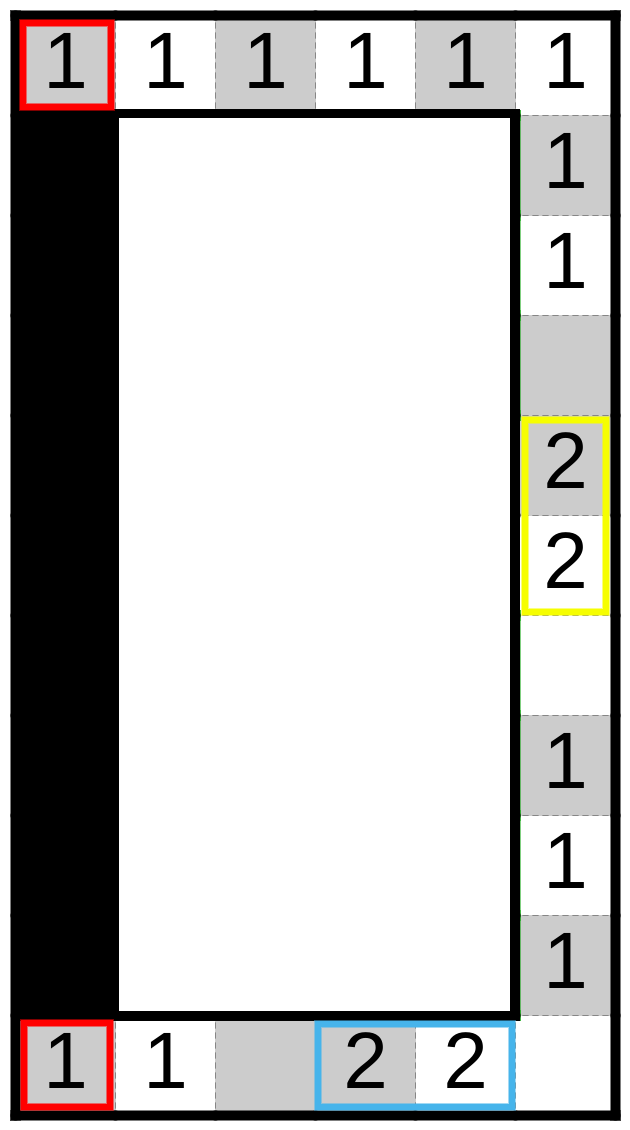}
                \caption[Network2]%
                {{\small An unsolved middle wire split subgadget. The black is the mandatory area of the crossover subgadget. The areas in blue, red, and yellow rectangles are the optional areas with the middle wire, the crossover, and the OR subgadgets respectively.}}    
                \label{fig:uwsplitter}
            \end{subfigure}
            \hfill
            \begin{subfigure}[t]{0.21\textwidth}  
                \centering 
                \includegraphics[scale = 0.15]{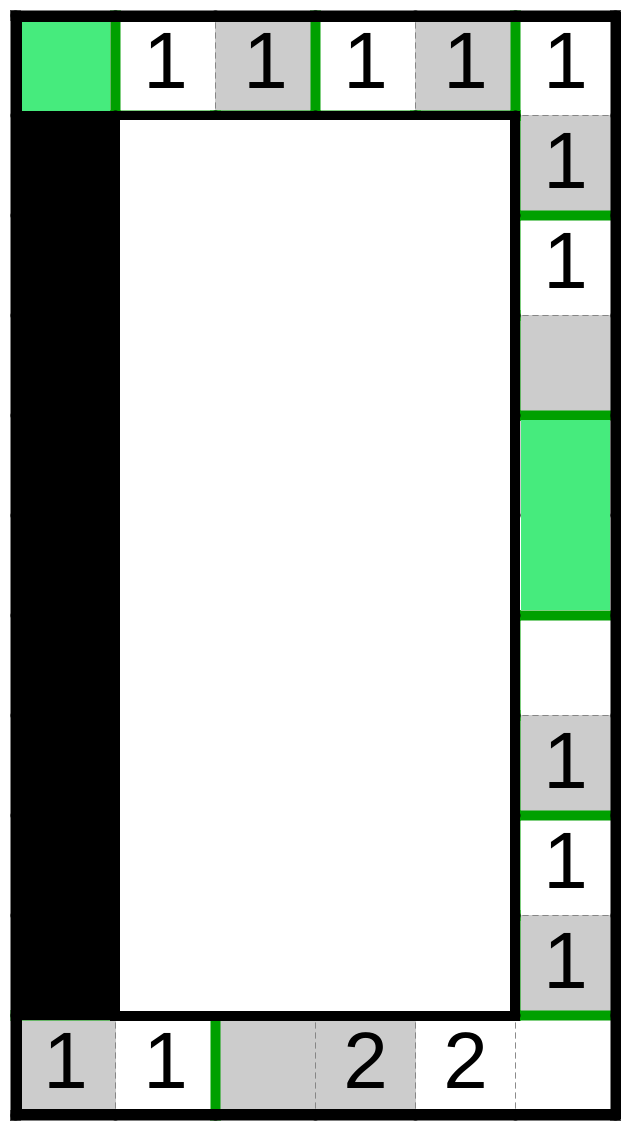}
                \caption[]%
                {{\small The split middle wire communicating \textit{false}.}}    
                \label{fig:fwsplitter}
            \end{subfigure}
            \hfill
            \begin{subfigure}[t]{0.21\textwidth}   
                \centering 
                \includegraphics[scale = 0.15]{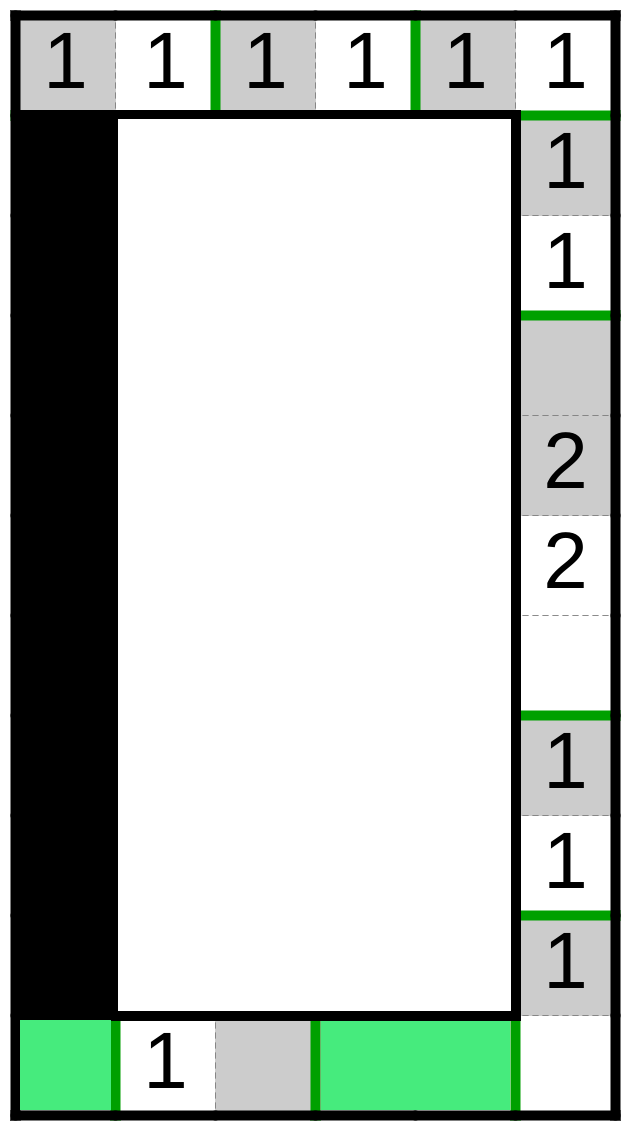}
                \caption[]%
                {{\small The split middle wire communicating \textit{true}.}}    
                \label{fig:twsplitter}
            \end{subfigure}
            \caption[]
            {\small  The middle wire split subgadget with solutions for the different values of the middle wire. The green ares are covered by adjacent gadgets.}
            \label{fig:wsplitter}
        \end{figure*}
        
        In \autoref{fig:wsplitter}, we can  see that the subgadget shares two optional areas with the crossover subgadget, while the right optional area represents the connection with the OR subgadget. The middle wire split subgadget will transfer the value of the middle wire from the blue optional area to the middle optional area of the OR subgadget (yellow optional area).
        \newpage
        \begin{lemma} \label{lem:wsplit}
            The middle optional area of the OR subgadget is covered by the solution of the middle wire split subgadget if, and only if, the middle wire communicates \textit{true}.
        \end{lemma}
        \begin{proof}
            If the middle wire has the value \textit{false}, then the blue optional area must be in the same block as the cells on its left and right. Hence, the lower red optional area will be in the same block as the cell on its right. Therefore, by Lemma \autoref{lem:cross}, the upper optional area is not covered by this solution. On the way inducing the solution shown in \autoref{fig:fwsplitter}.
        
            If the middle wire communicates the \textit{true} value, then the cell on the left of the blue optional area in the above figure must be in the same block as its only uncovered neighbour. Thus, leaving the lower red optional area uncovered by the gadget solution. Hence, by Lemma \autoref{lem:cross}, the upper red optional area must be covered by the solution of the middle wire split subgadget, forcing the solution shown in \autoref{fig:twsplitter}.
        \end{proof}
        
        With this in mind, we can prove that our clause gadget works as it should. 
        
        \begin{theorem} \label{thm:clause}
            The clause gadget is solvable if, and only if, at least one incoming wire communicates the \textit{true} value. 
        \end{theorem}
        \begin{proof}
            By Lemma \autoref{lem:cross}, the left optional area of the OR subgadget is covered by the crossover subgadget if, and only if, the left wire has the \textit{true} value. Lemma \autoref{lem:wsplit} implies that the middle optional area of the OR subgadget is covered by the solution of the middle wire split subgadget if, and only if, the middle wire communicates \textit{true}. That combined with Lemma \autoref{lem:OR} proves the statement.
        \end{proof}
        
        \newpage
         \begin{figure}[htp]
            {\hfill}
            \centering
            \begin{minipage}[b]{.45\textwidth}
                 \includegraphics[width=\textwidth]{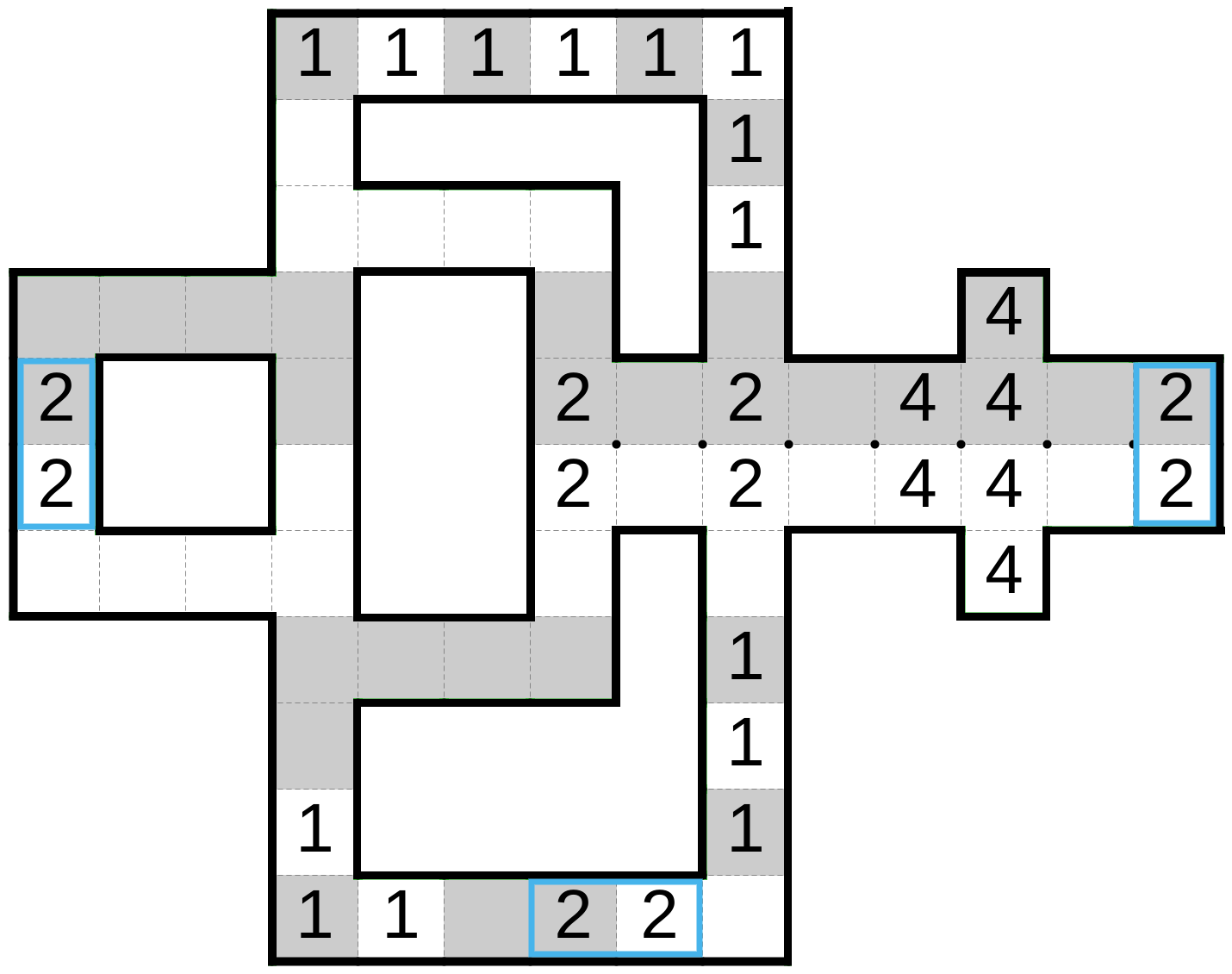}
                \caption[]
                {\small The clause gadget. The optional areas are in the blue rectangles, representing the places where the clause receives the value from the wires.}
                \label{fig:clause}
            \end{minipage}%
            {\hfill}
            \centering
            \begin{minipage}[b]{.5\textwidth}
               \centering
            \begin{subfigure}[t]{0.45\textwidth}
                \centering
                \includegraphics[scale = 0.15]{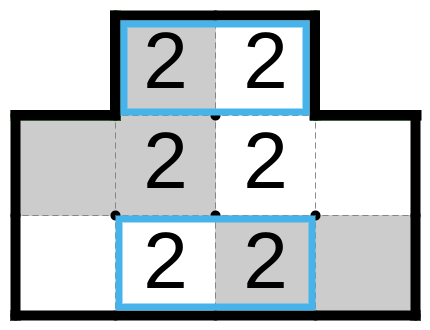}
                \caption[Network2]%
                {{\small An unsolved color inversion gadget. }}
                \label{fig:uinv}
            \end{subfigure}
            \hfill
            \begin{subfigure}[t]{0.45\textwidth}   
                \centering 
                \includegraphics[scale = 0.15]{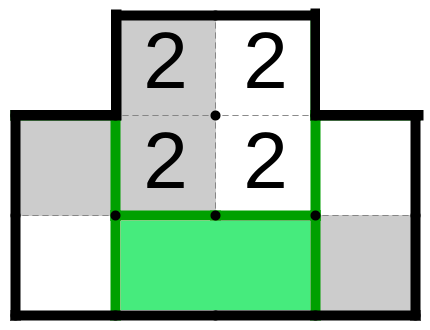}
                \caption[]%
                {{\small Inverting the color of the wire communicating communicating \textit{true}.}}    
                \label{fig:tinv}
            \end{subfigure}
            \vskip \baselineskip
            \begin{subfigure}[t]{0.45\textwidth}  
                \centering 
                \includegraphics[scale = 0.15]{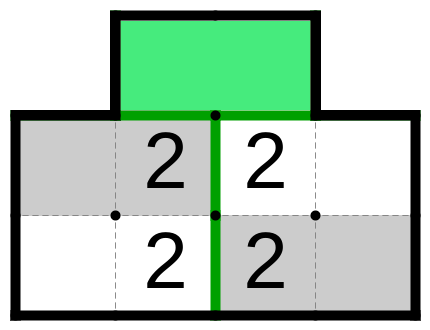}
                \caption[]%
                {{\small The first solution of inverting the color of the wire communicating \textit{false}.}}    
                \label{fig:finv-1}
            \end{subfigure}
            \hfill
            \begin{subfigure}[t]{0.45\textwidth}  
                \centering 
                \includegraphics[scale = 0.15]{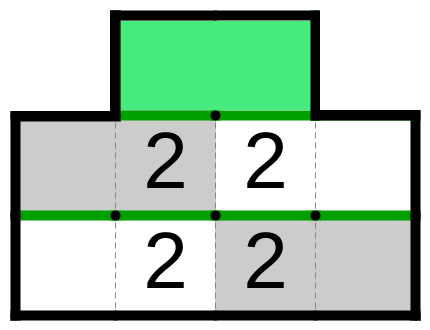}
                \caption[]%
                {{\small The second solution of inverting the color of the wire communicating \textit{false}.}}    
                \label{fig:finv-2}
            \end{subfigure}
            \caption[]
            {\small The color inversion gadget with the profile table. Paired optional areas are in blue rectangles. It receives the value on the lower optional area, while it outputs it on the upper one.}
            \label{fig:inv}
            \end{minipage}
            {\hfill}
        \end{figure}

        Even though our clause gadget has desired characteristics, there is one problem that still may occur. Some wires incident to the gadget may have the opposite coloring of the one in the optional areas. Although the clause gadget is invariant on the color inversion, the color inversion will change coloring of all the optional areas, which may shift the problem to other optional areas. Hence, we need a different solution.
    
        The solution to this problem is the \textit{color inversion} gadget, presented in \autoref{fig:inv}. This gadget will receive the truth value from the wire on its lower optional area, and the color inverted wire of the same value will then exit on its upper optional area.
        
        \begin{lemma} \label{lem:inv}
            The profile table of the color inversion gadget in \autoref{fig:inv} is \textit{sufficient}.
        \end{lemma}
        \begin{proof}
            It is easy to check that the only solutions where the optional areas are either covered or uncovered are shown in \autoref{fig:inv} (b) - (d).
        \end{proof}
        
        \begin{lemma} \label{lem:inv1}
            The lower optional area is the covered by the solution of the color inversion gadget if, and only if, the upper one is not covered.
        \end{lemma}
        \begin{proof}
            The proof immediately follows from \autoref{fig:inv} and Lemma \autoref{lem:inv}.
        \end{proof}
        The color inversion gadget will be required only when the wire arrives to the clause gadget, thus we can make the color inversion gadget the integral part of the clause gadget. We will call such gadget the \textit{complete clause gadget} (see \autoref{fig:comp-clause}). Since the clause gadget is invariant with respect to color change, we can let the middle wire determine the coloring of the gadget. Depending on the coloring of the left and right wire, we will either use or not use the color inversion gadgets.
        
        If the color of the wire should be changed, we will append the color inversion gadget surrounded from both sides with the wire gadgets, otherwise we will append the longer wire gadget. Observe that this is not necessary, but we will do it in order to ensure the width of the clause gadget remains unchanged. Newly adopted parts of the gadget will only transfer the value of the wire with the possible inversion of the wire's color.
        
        \begin{figure}[t]
            \centering
            \includegraphics[width=\textwidth]{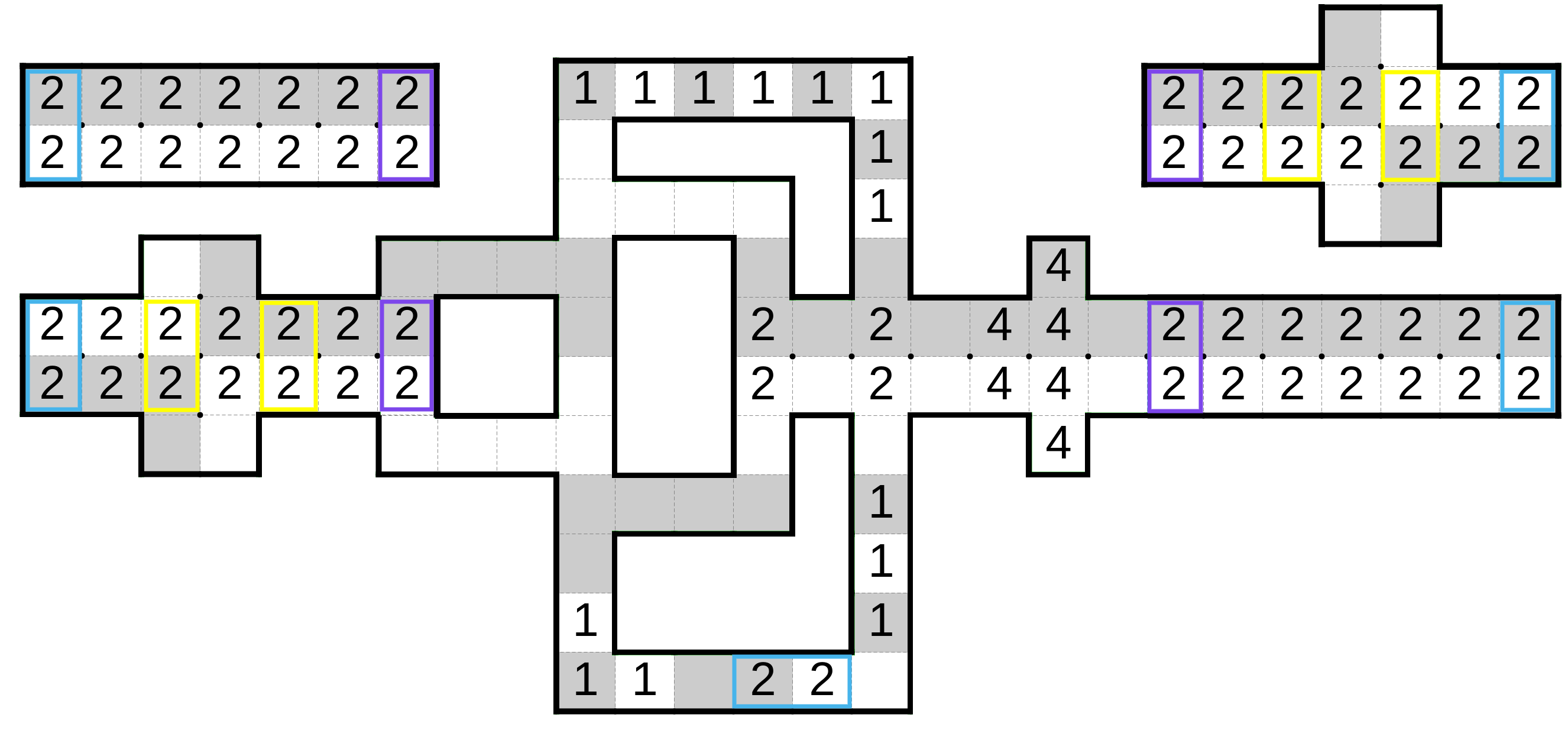}
            \caption[]
            {\small The complete clause gadget. The places where the wires are connected to the gadget are in blue rectangles. The optional areas shared between the color inversion gadgets and additional wire gadgets are shown in yellow rectangles. Area in the purple rectangles is the area where the newly added parts of the gadget are connected with the gadget from \autoref{fig:clause}. We can observe that color of the left wire is inverted while the color of the other wires remains unchanged. The two smaller segments not connected to the rest of the gadget are not the part of the gadget. They serve as possible replacements for the parts below them in the case that the left or the right wire have different coloring than the optional areas of the gadget. }
            \label{fig:comp-clause}
        \end{figure}

        \begin{lemma}
            The complete clause gadget is solvable if, and only if, at least one incoming wire communicates the \textit{true} value.
        \end{lemma}
        \begin{proof}
            The proof is an immediate consequence of the construction of the complete clause gadget, and \autoref{thm:clause}, Lemma \autoref{lem:wire}, and Lemma \autoref{lem:inv1}.
        \end{proof}
        
        Observe that the place where the left and the right wires enter the gadget are equally distanced from the middle optional area, in a way that emulates \autoref{fig:bent}.
        
    \newpage
    \section{Connecting the dots} \label{sec:conn-3}
        
        We will construct the puzzle starting from the instance of Bent RPM-3SAT (see \autoref{fig:bent}), and convert each of its parts to the equivalent gadget.
        
        First, we start from the variable line. We will place all the \textit{complete variable gadgets} on the same row of the puzzle, with enough space (even number of cells) between their closest cells, so that even the wires of the adjacent variables can be connected to the same clause. Next, we place the \textit{complete clause gadgets} according to the figure, either above or below the variable line, depending on whether the clause is positive or negative. The middle optional area of the \textit{complete clause gadgets} must be in the same columns as the optional area of the \textit{complete variable gadgets}, where their middle wires will exit. We widen the space between the clauses more than in the figure, to ensure that there will be an empty space between them, as well as with the wires. Finally, we draw the wires, starting from the \textit{complete variable gadgets}, using the \textit{wire gadgets} to connect it to the \textit{complete clause gadgets}. The middle wire can be directly connected to the clause. We will draw the left and the right wires of the clause until the wire's optional area is in the row below or above, for positive and negative clauses respectively, the optional area where it should enter the gadget. Then, we append the desired \textit{rotation gadgets} on top of the wire gadgets, and from there extend the wires enough to connect it with the \textit{complete clause gadgets}.
        
        \begin{lemma}
            All three wires of every \textit{complete clause gadget} can be connected to it. 
        \end{lemma}
        \begin{proof}
            In order to connect the middle wire of the clause with the \textit{complete clause gadget}, the optional area where the wire exits the \textit{complete variable gadget} must be an odd number of rows away from the optional area where it enters the \textit{complete clause gadget}. Hence, we will space them accordingly when constructing the puzzle. If we append the rotation gadget directly on the left or the right optional area of the \textit{complete clause gadget}, and we put the \textit{wire gadget} with the three rows of the mandatory area on the other side of the rotation gadget, the bottom optional area of the \textit{wire gadget} will be in the same row as the middle optional area of the \textit{complete clause gadget} while there will be 12 columns in between them. That distance can only be extended for an even number of columns, by adding another \textit{wire gadget} in between the \textit{rotation} and the \textit{complete clause gadgets}.  
            
            Because any two \textit{complete variable gadgets} have even distance between them, which we choose, and the horizontal distance between each wire until the end of its \textit{complete variable gadgets} is even, the distance between any two wires in the puzzle is even. 
            
            Hence, we can position the \textit{complete variable gadgets} so that all three wires of all of the \textit{complete clause gadgets} can be connected to them. 
        \end{proof}

        Since our construction follows the planar drawing, there would be no intersections of the gadgets. Because we created additional space between the gadgets, we forbid unwanted interference between non connected gadgets. 
        
        \begin{lemma} \label{lem:asdcNP}
            The arbitrary shaped \textit{Double Choco} is in the \textit{NP}.
        \end{lemma}
        \begin{proof}
            Given the certificate (the set of blocks) to the puzzle instance, we can validate it in the polynomial time by checking if the set of the blocks satisfies the \textit{Double Choco} properties from \autoref{chap:1}. 
        \end{proof}
        
        \begin{lemma}
            If the Bent RPM-3SAT has \textit{n} variable and \textit{m} clauses, then our reduction (puzzle) has the size polynomial in $n+m$ and it can be performed in polynomial time.
        \end{lemma}
        \begin{proof}
            The construction of the puzzle follows the Bent RPM-3SAT embedding, hence the height of the puzzle is $O(m)$.  Meanwhile, its width is $O(w)$, where $w$ represents the number of wires of the Bent RPM-3SAT. Since the number of edges in the planar graph is at most linear in the number of the vertices, we have that the width of our puzzle is $O(n + m)$. Hence, the size of our puzzle is polynomial in $n + m$. Translation of each element of the Bent RPM-3SAT (variable, wire or clause) to the homologous gadget and its inserting to the puzzle, takes the most polynomial time. Thus, our reduction is performed in the time polynomial in $n+m$.
        \end{proof}
        
        \begin{lemma} \label{lem:asdcNPhard}
            The Bent RPM-3SAT is solvable if, and only if, the puzzle is solvable.
        \end{lemma}
        \begin{proof}
            Given the solution to the Bent RPM-3SAT instance, we can configure the \textit{complete variable gadgets} to have the same value as their respective variables in the original problem. The \textit{wire} and the \textit{rotation} gadgets will transfer that value to the \textit{complete clause gadgets}. Since, in the original problem, each clause is connected to at least one literal with \textit{true} value, the \textit{complete clause gadgets} will be solvable because at least one of its adjacent \textit{wire gadgets} communicates \textit{true}.
            
            If the puzzle is solvable, then by the property of the \textit{complete variable gadget}, each \textit{wire gadget} carries either the \textit{true} or the \textit{false} value. That value is propagated to the \textit{complete clause gadgets} which are all solvable. Hence, the assigning of the variable's value of the original problem, according the to value of the respective \textit{complete variable gadgets}, provides at least one \textit{true} literal per the clause of the Bent RPM-3SAT.
        \end{proof}

        \begin{theorem}
            The arbitrary shaped Double Choco is NP-complete.
        \end{theorem}
        \begin{proof}
            Immediately from the previous three statements.
        \end{proof}

    \chapter{Double Choco is NP-complete} \label{chap:4}
        In this chapter we come back to our original problem. We will prove that the \textit{Double Choco} with the rectangular $m \times n$ board is \textit{NP-complete}, by using the results from \autoref{chap:3} and slightly modifying the gadgets, while leaving their properties unchanged, and filling the empty areas in the board. 
        
        Our gadgets will be surrounded by a thin layer, called a shield, which satisfies the following:
        \begin{itemize}
            \item The shield can be tiled by using solely its own cells.
            \item None of the shield's cells can be in the same block as any of the internal gadget cells in any puzzle solution.
        \end{itemize}
        
        These properties will ensure that any puzzle solution has a valid \textit{profile assignment}. If the puzzle was not solvable prior to shielding, the shield will not make it solvable.
        
        By shielding the gadgets, we prohibit undesired interference between the non connected gadgets, as well as interference between the external areas with the gadgets. We will present the shielded gadgets in \autoref{sec:swire}, \autoref{sec:svar}, and \autoref{sec:sclau}.
        
        Even now, it can happen that we have holes in our puzzle, between the shields of the different gadgets, thus, we need to fill them. That will be done by using fillers. In \autoref{sec:fholes}, we will argue that all the holes in our puzzle have the specific shape that will allow us to cover them whenever the puzzle is created according to the described procedure. Both the shields and the fillers are external areas (see \autoref{sec:sol}), meaning that they consist only of mandatory areas.
        
        Because each optional area belongs to the two gadgets, it must be shielded by only one of them, to avoid the shields overlapping. Hence, we decide to surround the optional areas with the shield of the gadget that receives the value through it, while the optional areas will not be shielded on the gadgets which outputs the value on it.
        
        All the shielded gadgets in this chapter can be extended in a similar way as in \autoref{chap:3} and additionally following the described shielding pattern.
        
        \newpage
       
        \begin{figure*}[t]
            \centering
            \begin{subfigure}[t]{0.3\textwidth}
                \centering
                \includegraphics[scale = 0.2]{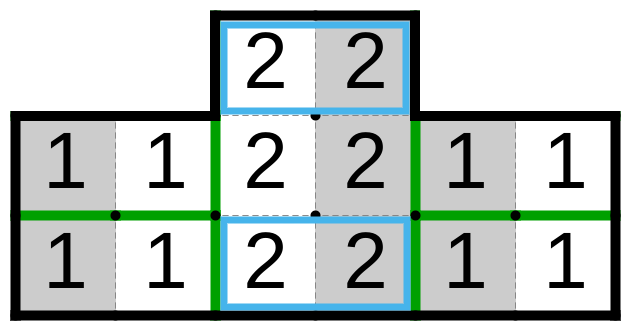}
                \caption[Network2]%
                {{\small Wire gadget with the shield.}}    
                \label{fig:s-wire}
            \end{subfigure}
            \hfill
            \begin{subfigure}[t]{0.3\textwidth}  
                \centering 
                \includegraphics[scale = 0.2]{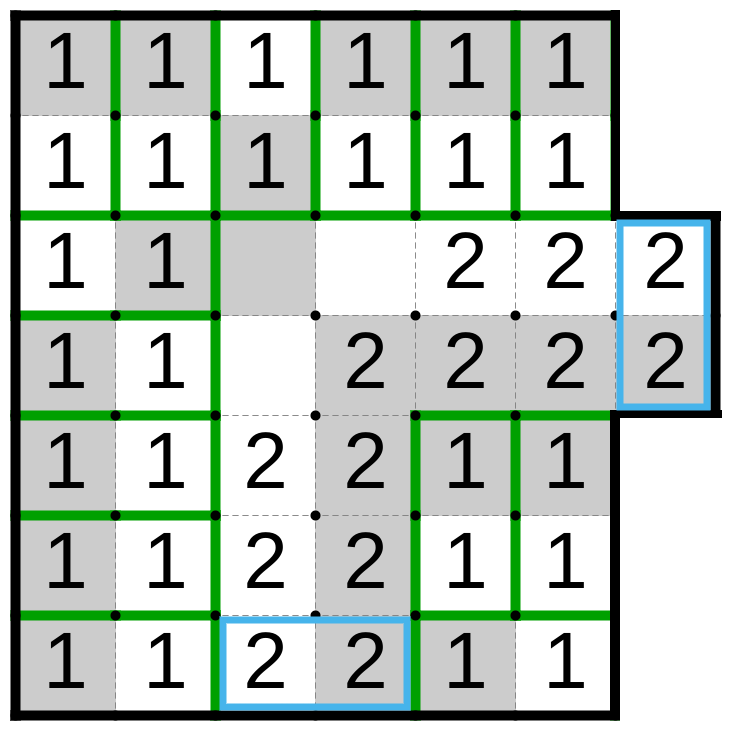}
                \caption[]%
                {{\small Right rotation gadget with the shield.}}    
                \label{fig:s-rangle}
            \end{subfigure}
            \hfill
            \begin{subfigure}[t]{0.3\textwidth}   
                \centering 
                \includegraphics[scale = 0.2]{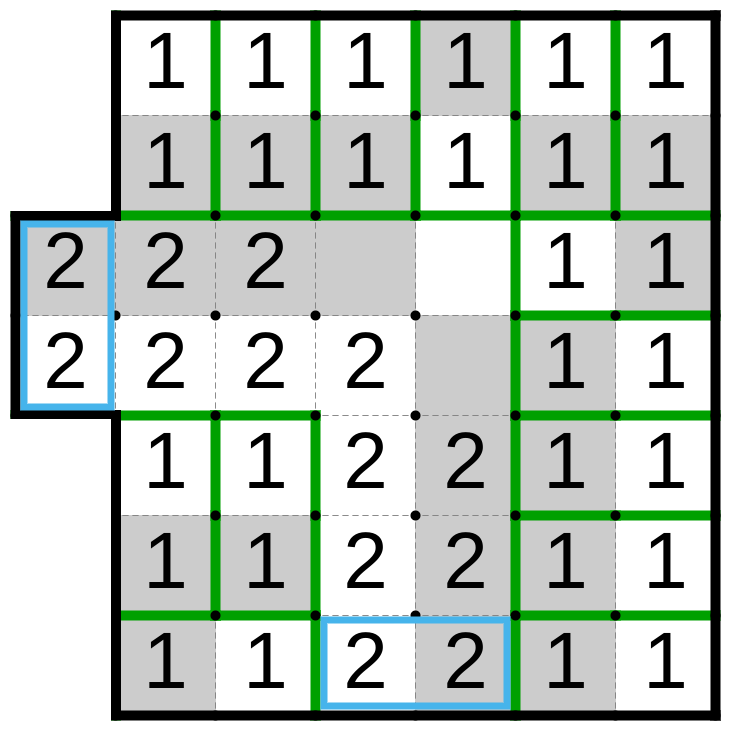}
                \caption[]%
                {{\small Left rotation gadget with the shield.}}    
                \label{fig:s-langle}
            \end{subfigure}
            \caption[]
            {\small  The wire and rotation gadgets with the shield. Shield is represented as the covered area. Optional areas are in the blue rectangles.}
            \label{fig:s-wires}
        \end{figure*}
        \section{Shielding the wire and rotation gadgets} \label{sec:swire}
        
        In \autoref{fig:s-wires}, we have presented the shielded version of the wire and rotation gadgets. The shield is represented with the areas containing an integer \textbf{1}, surrounding the gadgets everywhere except on the optional areas. The wire gadget stays unmodified, and, on the other hand, our rotation gadgets from \autoref{fig:angle} are extended with the wire gadgets appended to both of their optional areas, in order to shield them easier. Every gadget in the above figure has one optional area surrounded by the shield and one not. We will call them the input and output optional area respectively.
        
        \begin{claim}
            The shields of the gadgets in \autoref{fig:s-wires} can be tiled by solely using its own cells.
        \end{claim}
        \begin{proof}
            The solution is shown in the figure.
        \end{proof}

        \begin{lemma} \label{lem:swires}
            None of the shield's cells in the above figure, can be in the same block as the any of the gadget's cells.
        \end{lemma}
        \begin{proof}
            Any of the shield's cells are either filled with a different integer than its adjacent gadget's cells, or it has the same color as the adjacent gadget's cells, while having an integer \textbf{1} in it. Thus, by the \textit{Double Choco} properties, they cannot be in the same block.
        \end{proof}
        
        \begin{claim}
            The properties of shielded gadgets stay the same as in the ones in \autoref{sec:wire}. 
        \end{claim}
        \begin{proof}
            The statement follows from their construction and Lemma \autoref{lem:swires}.
        \end{proof}
         
        \newpage
        
        \begin{figure}[t]
            \centering
            \includegraphics[width=0.7\textwidth]{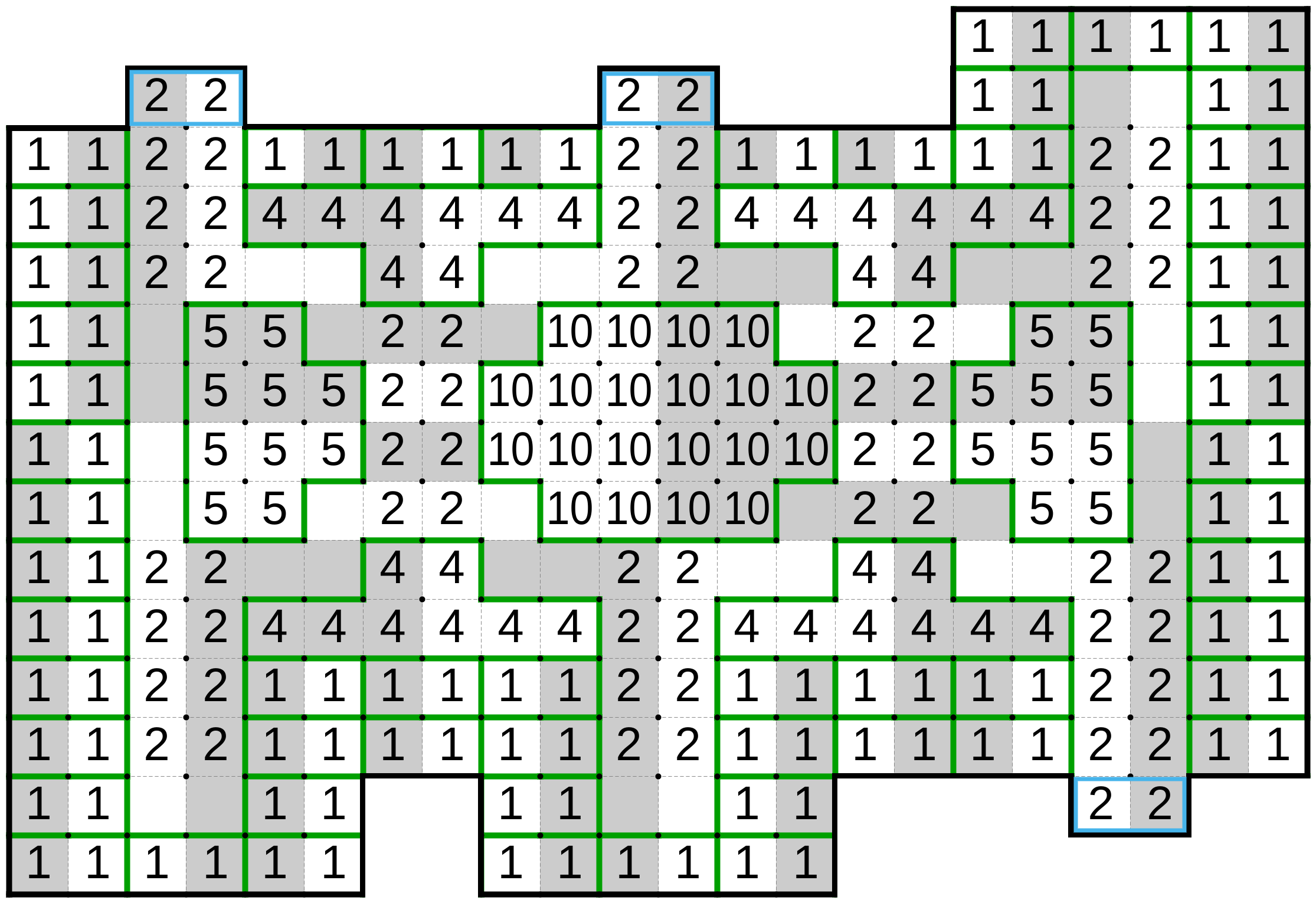}
            \caption[]
            {\small The complete variable gadget with the shield. The shield is represented by the covered area. Optional areas are in the blue rectangles.}
            \label{fig:s-var}
        \end{figure}
        
        \section{Shielding the complete variable gadget} \label{sec:svar}
        
        We present the shielded complete variable gadget in \autoref{fig:s-var} with the three wires, two going upwards and one going downward. The gadget stays mainly unmodified, except where we added additional wire gadgets on the lower side, in between the optional area of the variable gadget (see \autoref{fig:var}) and the wire boundary gadgets, or the wire gadgets, depending on the case. Observe, that the all optional areas are the output ones.
        
        Holes between the equalizers and the vertical lines, as well as the ones between the two equalizers are filled with the areas of the \textbf{5} and \textbf{10} integers respectively. Even though there cannot be any interference between the neighbouring gadgets, if we filled them arbitrarily we risk the ill behaviour of our gadget.
        
        \begin{claim}
            The shield of the gadgets in \autoref{fig:s-var} can be tiled by solely using its own cells.
        \end{claim}
        \begin{proof}
            The solution is shown in the figure.
        \end{proof}
        
        \begin{lemma} \label{lem:svar}
            None of the shield's cells in the above figure can be in the same block as the any of the gadget's cells.
        \end{lemma}
        \begin{proof}
            Every shield's cell with the integer \textbf{1} in it, if directly adjacent to the gadget, has the neighbour inside the gadget of the same color, thus, they cannot be in the same block.
            It is easy to check that any cell with the integers \textbf{4}, \textbf{5} or \textbf{10} cannot be in any other block different from the ones shown in \autoref{fig:s-var}.
        \end{proof}
        
        \begin{claim}
            The properties of the shielded gadget stays the same as its counterpart in \autoref{sec:variable}.
        \end{claim}
        \begin{proof}
             The statement follows from the construction and Lemma \autoref{lem:svar}.
        \end{proof}

        \newpage
        
        \begin{figure}[t]
            \centering
            \includegraphics[width=0.8\textwidth]{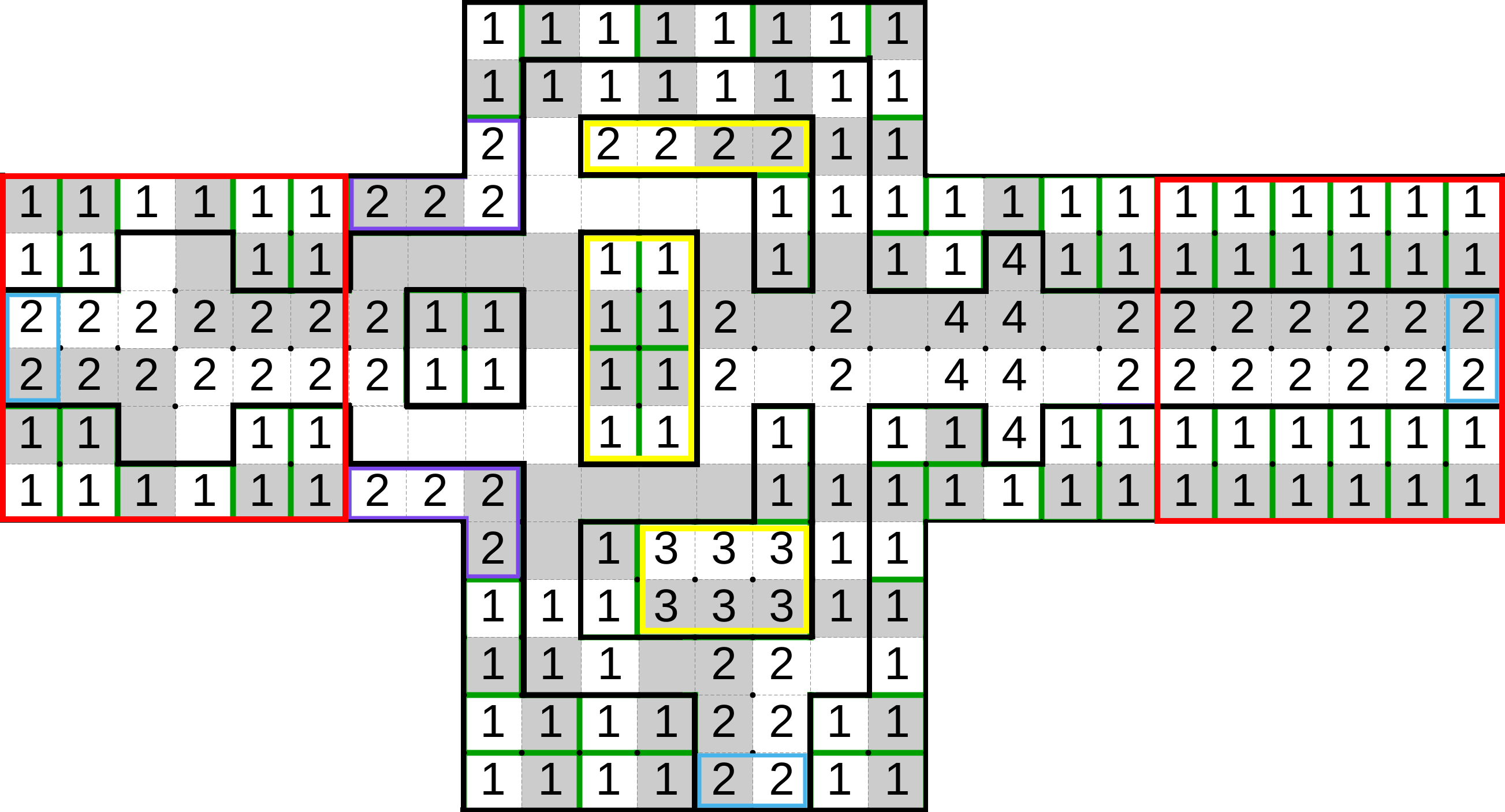}
            \caption[]
            {\small The complete clause gadget with the shield. The gadget is separated from the shield by the black border. Shield is represented as the covered area. The optional areas are in the blue rectangles. Areas in the red are those that depend on the color of the wire. Therefore, on both of those areas we can use either one of them.}
            \label{fig:s-clause}
        \end{figure}
        
        \section{Shielding the complete clause gadget} \label{sec:sclau}

        In \autoref{fig:s-clause}, we present the complete clause gadget with the shield. As expected, the gadget has the only input optional areas. The only modified part of the gadget in \autoref{fig:comp-clause}, is the middle optional area where we extend the gadget by appending an additional wire gadget on its optional area. The gadget with the shield misses the four rectangular areas in order to have rectangular shape. We omitted presenting those rectangles in the figure in order to simplify it. We will prove in the next claim that they can be filled to satisfy the shield's properties. Hence, from now on, we will assume that the complete clause gadget with the shield has rectangular shape.
        
         \begin{claim}
            The shields of the complete clause gadget can be tiled solely using its own cells.
        \end{claim}
        \begin{proof}
            Part of the solution is shown in \autoref{fig:s-clause}. For the four rectangular areas not shown in the figure, observe that the width of all four of them is even, hence by Lemma \autoref{lem:rectangles}, they can be filled with $2 \times 1$ disjointed rectangles. We can make each rectangle have one white and one gray cell with the integer \textbf{1} in them, and they will satisfy the properties of the shield.
        \end{proof}
        
        \begin{lemma} \label{lem:sclause}
            If the gadget and the shield in the above figure has a solution, than none of the shield's cells can be in the same block as any of the gadget's cells in any solution of the puzzle.
        \end{lemma}
        \begin{proof}
            Any shield's cell, except some in the yellow and the purple rectangle, have the integer \textbf{1} in it, and if they it is directly adjacent to the gadget, all its neighbours inside the gadget have the same color or different integer in it, thus, they cannot be in the same block.
            
            The white cells of the upper yellow area can be tiled only with gray cells of the same area. Furthermore, the gray cells in the lower yellow area must be in the same block as its corresponding white cells. 
            
            The upper purple area is surrounded with cells with the integer \textbf{1} on all sides except for the interior of the gadget. Because the the gray cells of the area must be in the same block as the lower white cell, the whole area will be in the same block. Otherwise, the upper white cell will remain uncovered and it cannot be tiled on the any other way. The same argument applies for the lower purple area.
            
            We are left to prove that the middle yellow area must be locally solved. Observe that its cells in the second row must be covered with the cells above them, also the right cells in the third row must be in the same block as the cell below it. Hence, we have to prove that the left cells in the third and the fourth row must be in the same block, and we will do so by contradiction. Assume that they are not in the same block, then the left cell in the third row is tiled with the cell on its left, while the cell in the fourth row is in the same block as the cell below it. But, then, the white area left of the bottom row of this area cannot be tiled. Hence, we have reached the contradiction.
        \end{proof}
        
        \begin{claim}
            The properties of shielded gadget stays the same as its counterpart in \autoref{sec:clause}.
        \end{claim}
        \begin{proof}
             The statement follows from the construction and Lemma \autoref{lem:sclause}.
        \end{proof}

        \section{Connecting the dots} \label{sec:fholes}
            We will follow the same construction as in \autoref{sec:conn-3}, and increase the space between the gadgets to make room for the shields. Let us assume that the leftmost column of the leftmost complete variable gadget is the first column of our puzzle, while the rightmost column of the rightmost complete variable gadget is the last column of the puzzle. Observe that, because those gadgets have even width, and they are evenly spaced, our puzzle also has even width.
            
            Now we will prove that the remaining holes in our puzzle can be filled by the \textit{filler areas}.
            
            \begin{lemma} \label{lem:rectangles}
                Every rectangle with an even number of cells in at least one dimension can be represented as the union of the disjointed $2 \times 1$ rectangles.
            \end{lemma}
            \begin{proof}
                Let's suppose that the rectangle's have an even number of columns, a similar argument applies when the number of rows is even. Take one row, it has an even number of cells, so we can split it into $2 \times 1$ rectangles. Doing the same for the each row, we get a desired set of rectangles.
            \end{proof}
            
            \begin{theorem} \label{thm:space}
                Every empty area of our puzzle can be separated into the union of the disjointed rectangles with the even width.
            \end{theorem}
            \begin{proof}
                We will argue that the space between any of the gadgets in our construction has even width. First observe that every shielded gadget has even width. Since our construction follows \autoref{fig:bent}, we can observe that horizontally adjacent elements of the graph can be:
                \begin{enumerate}
                    \item Variable with another variable - by our construction every two shielded complete variable gadgets are distanced horizontally by an even number of cells. 
                    \item Wire with another wire - By the previous item, any two complete variable gadgets are evenly spaced. Additionally, any wire gadget is the even number of columns away of the vertical borders of its complete variable gadget (see \autoref{fig:s-var}). Thus, any two wire gadgets are distanced horizontally by an even number of cells. 
                    \item Clause with another clause - by the construction the middle wire gadget will be connected directly to the complete clause gadget. By the previous item, the middle wires of the different clauses are evenly spaced. Furthermore, the middle optional area of the complete clause gadget is an even number of cells horizontally away from both the left and right edges of the shielded complete clause gadget (see \autoref{fig:s-clause}). Thus, the horizontal distance between any shielded wire and complete clause gadget is even.
                    \item Wire with the clause - the wire gadget has an even distance to the middle wire of the corresponding complete shielded clause gadget, hence its distance to the complete shielded clause gadget is even as well.
                \end{enumerate}
                In the previous enumeration we exclude rotation gadgets because they can have the horizontal empty area only on the one side, opposite of the rotation, but on that side, it has the same horizontal coordinates as its input wire gadget. Hence, we looked at it as the part of the wire gadget.
                
                Because the puzzle is bounded it has two more types of empty areas: from the vertical edge of the puzzle to the the shielded wire gadget, or the complete shielded clause gadget. But those cases are similar to the second and the fourth item in the previous enumeration, since each border is touching the shielded complete variable gadget, which has even width and even horizontal distance to all shielded wire and complete shielded clause gadgets. Thus, those horizontal distances are even as well.
            \end{proof}
            
            \begin{lemma} \label{lem:filler}
               Empty areas of our puzzles can be filled with the \textit{filler areas}.
            \end{lemma}
            \begin{proof}
                Combining \autoref{thm:space} with Lemma \autoref{lem:rectangles}, we get that the empty areas of the puzzle can be represented as the union of disjointed $2 \times 1$ rectangles. Let each rectangle have one white and one gray cell, while both of them are labeled by the integer \textbf{1}. Then, each rectangle can be tiled using solely its own cells.
            \end{proof}
            
            \begin{theorem}
                Double Choco in NP-complete.
            \end{theorem}
            \begin{proof}
                Since the shielded gadgets satisfied all the properties of the matching gadgets in \autoref{chap:3}, all the statements from \autoref{sec:conn-3} holds here as well. And by Lemma \autoref{lem:filler} we can fill all the holes in the puzzle. Hence, we conclude the proof.
            \end{proof}
        
    \newpage
    \hfill \textcolor{gray}{\textit{'It's not about the destination, it's about the journey.'}}
    
    \hfill \textcolor{gray}{R. W. E.}
    \newpage

    \printbibliography[heading=bibintoc, title={Bibliography}]

\end{document}